\documentclass[10pt, conference, compsocconf]{IEEEtran}
\IEEEoverridecommandlockouts
\usepackage{cite}
\usepackage{amsmath,amssymb,amsfonts}
\usepackage{graphicx}
\graphicspath{{Figures/}}
\usepackage{textcomp}
\usepackage{xcolor}
\usepackage{url}
\usepackage{algorithm}
\usepackage{algpseudocode}
\def\BibTeX{{\rm B\kern-.05em{\sc i\kern-.025em b}\kern-.08em
    T\kern-.1667em\lower.7ex\hbox{E}\kern-.125emX}}
\usepackage{enumitem}
\newtheorem{definition}{Definition}\newtheorem{theorem}{Theorem}
\newtheorem{proof}{Proof}
\begin{document}

\title{Explaining Aggregates for Exploratory Analytics
}

\author{\IEEEauthorblockN{Fotis Savva}
\IEEEauthorblockA{\textit{University of Glasgow, UK}\\
f.savva.1@research.gla.ac.uk}
\and
\IEEEauthorblockN{Christos Anagnostopoulos}
\IEEEauthorblockA{\textit{University of Glasgow, UK}\\
christos.anagnostopoulos@glasgow.ac.uk}
\and
\IEEEauthorblockN{Peter Triantafillou}
\IEEEauthorblockA{\textit{University of Warwick, UK}\\
p.triantafillou@warwick.ac.uk}
}

\maketitle

\begin{abstract}
Analysts wishing to explore multivariate data spaces, typically pose queries involving selection operators, i.e., range or radius queries, which define data subspaces of possible interest and then use aggregation functions, the results of which determine their exploratory analytics interests. However, such aggregate query (AQ) results are simple scalars and as such, convey limited information about the queried subspaces for exploratory analysis. We address this shortcoming aiding analysts to explore and understand data subspaces by contributing a novel explanation mechanism coined XAXA: eXplaining Aggregates for eXploratory Analytics. XAXA's novel AQ explanations are represented using functions obtained by a three-fold joint optimization problem. Explanations assume the form of a set of parametric piecewise-linear functions acquired through a statistical learning model. 
A key feature of the proposed solution is that model training is performed by only monitoring AQs and their answers on-line. 
In XAXA, explanations for future AQs can be computed without any database (DB) access and can be used to further explore the queried data subspaces, without issuing any more queries to the DB. 
We evaluate the explanation accuracy and efficiency of XAXA 
through theoretically grounded metrics over real-world and synthetic datasets and query workloads.
\end{abstract}

\begin{IEEEkeywords}
Exploratory analytics, predictive modeling, explanations, learning vector quantization, query-driven analytics.   
\end{IEEEkeywords}

\section{Introduction}
In the era of big data, analysts wish to explore and understand data subspaces in an efficient manner. The typical procedure followed by analysts is rather ad-hoc and domain specific, but invariantly includes the fundamental step of \textit{exploratory analysis} \cite{IPC15}. 
%Exploring data spaces is central for testing hypotheses, building predictive models, modeling data trends etc. %Once this step is complete, the possibilities are endless: from models that predict markets or give recommendations to models that decide whether a tumor is benign or malignant. 

Aggregate Queries (AQs), e.g., \texttt{COUNT}, \texttt{SUM}, \texttt{AVG}, play a key role in exploratory analysis as they summarize regions of data. Using such aggregates, analysts decide whether a data region is of importance, depending on the task at hand. However, AQs return scalars (single values) possibly conveying little information. For example, imagine checking whether a particular range of zip codes is of interest, where the latter depends on the count of persons enjoying a `high' income: if the count of a particular subspace is 273, what does that mean? If the selection (range) predicate was less or more selective, how would this count change? What is the data region size with the minimum/maximum count of persons? Similarly, how do the various subregions of the queried subspace contribute to the total count value and to which degree?
%that are often up to interpretation by the user. Meaning that the user does not know why the result is what it is or how it was generated. For instance, when obtaining the mean household income over a specific region, we do not know why the obtained result is high or low, which households drive the mean up or down or even how this result is affected by the size of the region. 
To answer such exploratory questions and enhance analysts' understanding of the queried subspace, one need to issue more queries. The analyst has no understanding of the space that would steer them in the right direction w.r.t. which queries to pose next; as such, further exploration becomes ad-hoc, unsystematic, and uninformed. 
To this end, the main objective of this paper is to find ways to assist analysts understanding and analyzing such subspaces. A convenient way to represent how something is derived (in terms of compactly and succinctly conveying rich information) is adopting regression models. In turn, the regression model can be used to estimate query answers. This would lead to fewer queries issued against the system saving system resources
and drastically reducing the time needed for exploratory analysis.

\subsection{Motivations and Scenarios}
We focus on AQs with a \textit{center-radius selection} operator (CRS) because of their wide-applicability in analytics tasks and easy extension to 1-d range queries. A CRS operator is defined by a multi-dimensional point (center) and a radius. Such operator is evident in many applications including: location-based search, e.g searching for spatially-close (within a radius) objects, such as astronomical objects, objects within a geographical region etc. In addition, the same operator is found within social network analytics, e.g., when looking for points within a similarity distance in a social graph. Advanced analytics functions like variance and statistical dependencies are applied over the data subspaces defined by CRS operators. 
% Last, they are natural to formulate using a user interface, which allows users to draw circles/hyper-spheres for selecting a region of interest within the visualized data space. Thus, allowing their use in Interactive Data Exploration. 
% The CRS operator can easily be extended to formulate an AQ over a specific attribute, e.g., a user issues an AQ involving the \texttt{VAR} and a CRS to retrieve the variance for a specific attribute, over the whole range or a portion of it.

\textbf{Scenario 1: Crimes Data.} Consider analyzing crimes' datasets,
containing recorded crimes, along with their location and the type of crime (homicide, burglary, etc.), like the Chicago Crimes Dataset \cite{crimesdata}. A typical exploration is to issue the AQ with a CRS operator: 
\begin{verbatim}
SELECT COUNT(*) AS y FROM Crimes AS C
WHERE $theta>sqrt(power($X-C.X,2)
+power($Y-C.Y,2));
\end{verbatim}
The multi-dimensional center is defined by (\texttt{\$X}, \texttt{\$Y}) and radius by \texttt{\$theta}. Such AQ returns the number of crimes in a specific area of interest located around the center defining a data subspace corresponding to an arbitrary neighborhood. 
Having a regression function as an explanation for this AQ, the analyst could use it to further understand and estimate how the queried subspace and its subregions contribute to the aggregate value. For instance, Figure \ref{fig:vis-rate} (left) depicts the specified AQ as the out-most colored circle, where \textit{(Lat, Lon)} blue points are the locations of incidents reported from the Chicago Crimes Dataset. The different colors inside the circle denote the different rates at which the number of crimes increases as the region gets larger. Thus, consulting this visualisation the analyst can infer the contributions of the different subspaces defined as the separate concentric circles. This solves two problems coined the \textit{query-focus problem} and the \textit{query-space size problem}. The query-focus problem denotes an inappropriate central focus by the multi-dimensional center. In the case at hand, a low count for the smallest concentric circle denotes a center of low importance and that the analyst should shift focus to other regions. In contrast, a low count for the largest concentric circle would mean an inappropriate original \texttt{\$theta} and that the analyst should continue exploring with smaller \texttt{\$theta} or at different regions.
% The different concentric circles and the different rates ( different colors at different radii) are obtained using a Piecewise-Linear Regression (PLR) function.
As will be elaborated later, an explanation regression function will facilitate the further exploration of the aggregate value for subregions of smaller or greater sizes without issuing additional AQ (SQL) queries, but instead simply plugging different parameter values in the explanation function.
% Specifically, given the original AQ from the analyst and a result $y$ the analyst could use the explanation function to infer the parameter \texttt{\$theta} for which the result $y'$ is a fraction of the original result, $y'=ay ,\quad a=[0.0,1.0]$. 
% For instance, a data scientist, in a given metropolis, is concerned with a specific district or even just a number of areas, thus the queries issued will be centered around those areas, filtering out data that are not within a particular region and then counting the number of results using a COUNT aggregate.
%An aggregate is an operator given by common RDBMS systems and include, but are not limited to, operators such as \texttt{COUNT}, \texttt{AVG}, \texttt{MEDIAN}. An example of such a query is in SQL syntax:

\textbf{Scenario 2: Telecommunication Calls Data.} Consider a scientist tasked with identifying time-frames having high average call times. She needs to issue AQs for varying-sized ranges over time, such as this SQL range-AQ. 
\begin{verbatim}
SELECT AVG(call_time) AS y FROM TCD AS C
WHERE C.time BETWEEN 
$X-$theta AND $X+$theta
\end{verbatim}
%(We can immediately see the usefulness to more than just spatial queries using a 1-d range selection operator on any attribute.) 
Discovering the aforementioned time-frames, without our proposed explanations, can be a daunting task as multiple queries have to be issued, overflowing the system with a number of AQs. 
% These could take minutes or hours to execute depending on the data size and throughput of the system. But, using the explanation function proposed in this paper, the analyst would be able to carry out her task with highly accurate answers by simply plugging different parameters to the given explanation function.
%one can infer the maximum (or minimum) points around a given AQ. 
Beyond that, analysts could formulate an optimization problem that could be solved using our methodology. 
Given a function, the maxima and minima points can be inferred, thus, the analyst can easily discover the parameter at which the AQ result is maximized or minimized. Again, such functionality is very much lacking and are crucial for exploratory analytics.
%Both tasks are impossible with just the answers to the AQs and entirely possible with an explanation function.
%As computer scientists we can call that turnaround time between when a case is submitted until the point it is updated.
% \begin{verbatim}
% SELECT AVG(call_time) AS z
% FROM Calls AS C
% WHERE $theta>sqrt(power($X-C.X,2)+power($Y-C.Y,2));
% \end{verbatim}

% Such examples extend to many different types of aggregates -- in this paper we focus on \texttt{COUNT, AVERAGE, SUM} since these are core functions noted in most exploratory work-flows. They also seem to be the initial focus of other notable publications \cite{agarwal2013blinkdb,VerdictDB}.  Nonetheless our framework can also provide explanations for other types of aggregates, such as \texttt{MIN, MAX, CORR}; we leave their discussion to future work, mostly for space reasons. The same motivations easily extend to other application domains.

\subsection{Problem Definition}
\label{sec:prob_def}
We consider the result of an AQ stemming from a regression function, defined from this point onwards as the \textit{true} function. With that true function fluctuating as the parameter values of an AQ (e.g \texttt{\$X, \$Y, \$theta}) change. Thus, we are faced with the problem of approximating the true function with high accuracy. Using the aforementioned facts we construct an optimization problem 
to find an optimal approximation to the \textit{true} function. In turn, this problem is further deconstructed into two optimization problems. The first being the need to find optimal query parameters representing a query space. The described AQs often form clusters as a number of AQs are issued with similar parameters. This is evident from Figure \ref{fig:vis-rate} (right) displaying AQs (from real-world analytical query workload \cite{szalay2002sdss}) forming clusters in the query space. AQs clustered together will have similar results, thus similar true functions. Consequently, we are interested in finding a set of optimal parameters to represent queries in each one of those clusters, as this will be more accurate than one set of global parameters\footnote{This basically represents all queries in a cluster by one query.}. The second part of the deconstructed optimization problem is approximating the true function(s) using a known class of functions over the clustered query space. To this end, we use piecewise-linear functions that are fitted over the clustered query space. In addition, we consider both historically observed queries and incoming queries to 
define a third joint optimization problem which is solved in an on-line manner by adapting both the obtained optimal parameters and the approximate explanation function.     

\begin{figure}
\begin{center}
\begin{tabular}{cc}
\includegraphics[height=4cm,width=5cm]{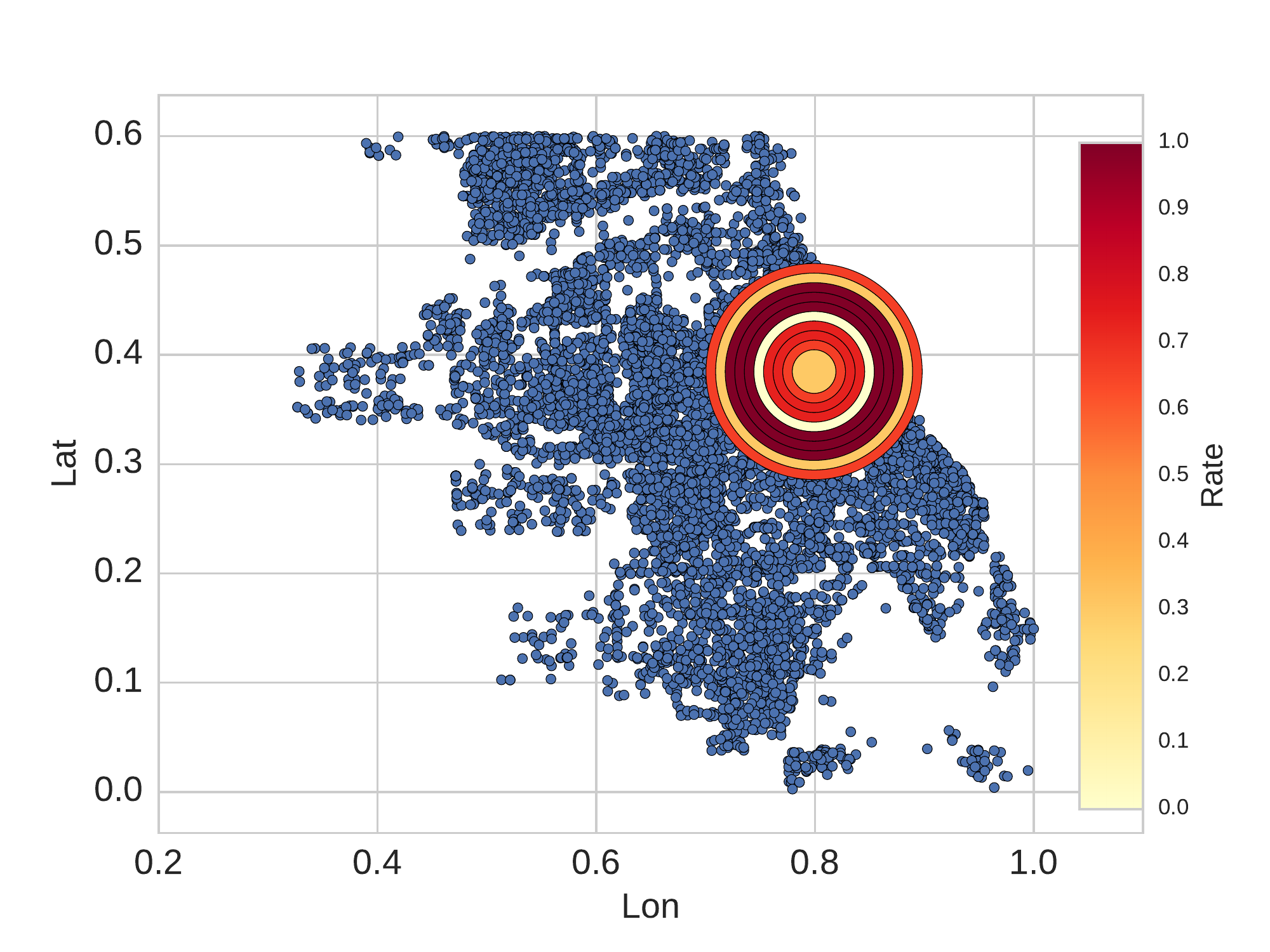}
\includegraphics[height=4cm,width=4cm]{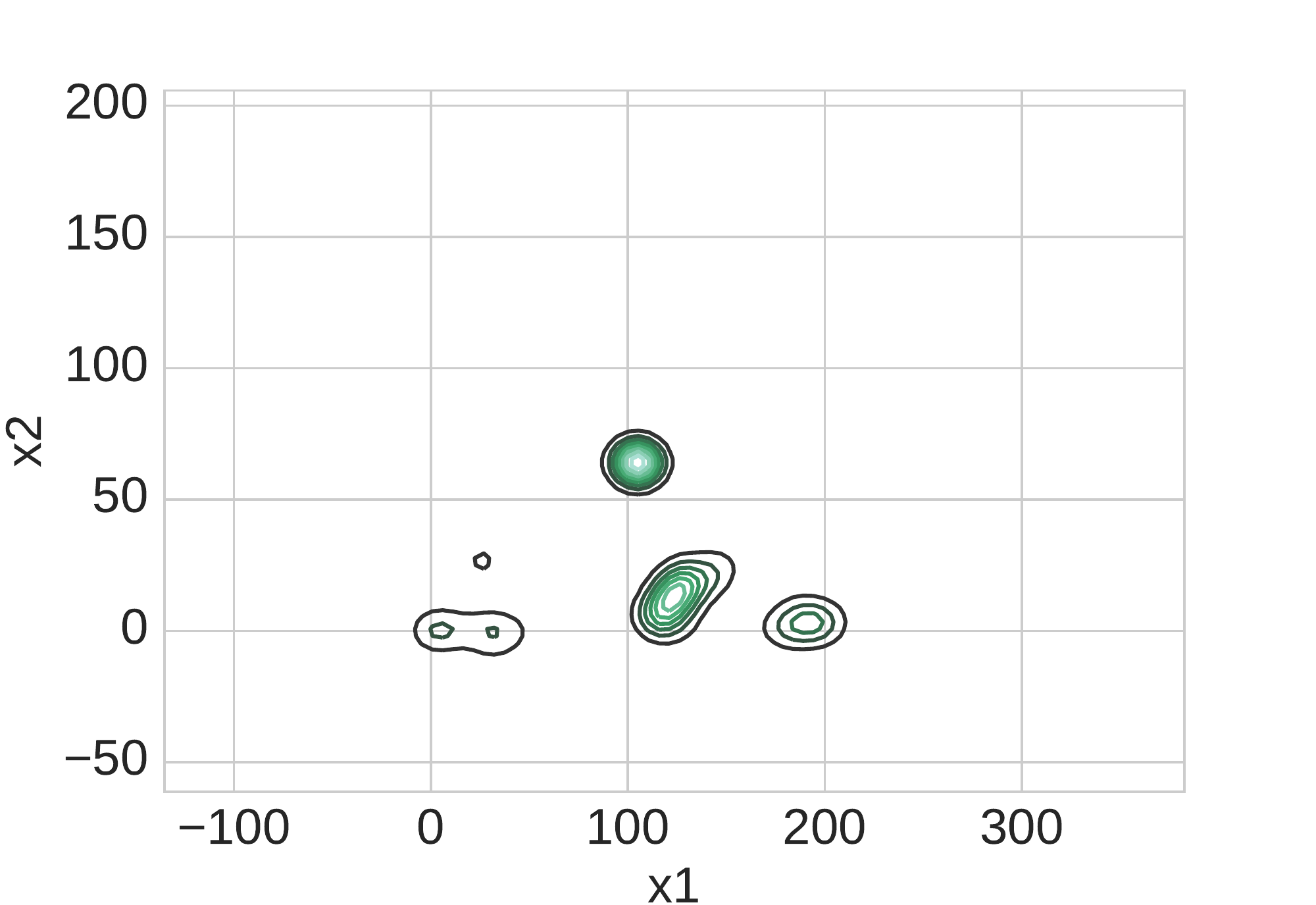}
\end{tabular}
\caption{(Left) Scenario 1: the blue points are locations of reported incidents; the circle on right is the AQ with a CRS operator with the varying rate-of-increase shown as color coded concentric circles; 
(right) Real workload cluster analysis (Source SDSS \cite{szalay2002sdss}); $x_{1}$ and $x_{2}$ are parameters of a CRS operator.}
\label{fig:vis-rate}
\end{center}
\end{figure}
\section{Related Work \& Contribution}

XAXA aims to provide explanations for AQ, whose efficient computation has been a major research interest \cite{HHW97, CDS04,CM05,WOT10, HBY13, AT17, AT15, WWDI17, agarwal2013blinkdb, VerdictDB}, with methods applying sampling, synopses and machine learning (ML) models to compute such queries.
Compared to XAXA, the above works are largely complementary. %For example, the data-gnostic method of XAXA can be used in tandem with any of the above approaches to replace expensive base data accesses and the data-agnostic method of XAXA can  expedite the training of ML models. 
However, the distinguishing feature of XAXA is that its primary task is to \textit{explain the AQ results} and do so efficiently, accurately, and scalably w.r.t. increasing data sizes. Others also focus on the task of assisting an analyst by building systems for efficient visual analytics systems \cite{vartak2015s,demiralp2017foresight}, however we note that our work is different in that it offers explanations for aggregates (as functions) that can be visualised (as the aforementioned works) but can also be used for optimization problems and for answering subsequent queries.

Explanation techniques have emerged in multiple contexts within the data management and ML communities. Their origin stems from data provenance \cite{cheney2009provenance}, but have since departed from such notions, and focus on explanations in varying contexts. One of such contexts is to provide explanations, represented as predicates, for simple query answers as in \cite{el2014interpretable, roy2015explaining,meliou2014causality}.
% Similarly, other authors have extended this in probabilistic and scientific databases \cite{kanagal2011sensitivity}, \cite{wu2013subzero}.
Explanations have been also used for interpreting outliers in both \textit{in-situ} data \cite{wu2013scorpion} and in streaming data \cite{bailis2016macrobase}. 
% The authors first detect outliers, either manually or automatically, and then generate predicates or attribute-value combinations that explain the outliers set. 
In \cite{nair2015learning}, the authors build a system to provide interpretations for service errors and present the explanations visually to assist debugging large scale systems. In addition, the authors of PerfXPlain \cite{khoussainova2012perfxplain} created a tool to assist users while debugging performance issues in Map-Reduce jobs. Other frameworks provide explanations used by users to locate any discrepancies found in their data \cite{wang2015data}, \cite{chalamalla2014descriptive}, \cite{wang2016qfix}. A recent trend is in explaining ML models for debugging purposes \cite{krishnan2017palm} or to understand how a model makes predictions \cite{sundararajan2017axiomatic} and conveys trust to users using these predictions \cite{ribeiro2016should}.

% Given the above, one can detect the central themes emerging around the concept of explanations. When working with explanations one has to determine the (i) domain, (ii) the representation, (iii) the approach that will be used and one has to consider (iv) the efficiency of generating explanations. For instance, in \cite{bailis2016macrobase} the domain is in outliers analysis, the explanations are represented using attribute-value combinations, and the approach is to use statistical structures that allow the analysis of streaming values. 
%Therefore, we proceed with these four central themes in mind, in order to build our explanations over aggregate queries.

In this work, we focus on explanations for AQs because of the AQs' wide use in exploratory analytics \cite{WWDI17}. Both \cite{wu2013scorpion} and \cite{amsterdamer2011provenance} focused on explaining aggregate queries. 
% However, the former focused on explaining the existence of outliers in aggregate queries and the latter on tracking the `how?' provenance of AQs using semiring formalisms. 
Our major difference is that we explain these AQs solely based on the input parameters and results of previous and incoming queries, thus not having to rely on the underlying data, which makes generating explanations slow and inefficient for large-scale datasets. 
% We wish to find ways to explain AQs and provide further insights to data analysts for establishing informed exploratory and descriptive statistics tasks, in an efficient manner. Hence, using the formalisms from \cite{amsterdamer2011provenance} is not possible as they would be incomprehensible and in need of provenance query language, as explicitly stated in \cite{amsterdamer2011provenance}.

% Our work is unique in developing and adapting ML models in order to increase efficiency/scalability when generating explanations. Providing accuracy while being efficient and scalable is being increasingly recognized as central within the community \cite{J13, BB07,J13}. 
%A popular approach rests on algorithmic weakening i.e., given higher volumes of data, faster (but less robust) algorithms are used where the increased data sizes can compensate for any loss in accuracy . 
% We propose a novel approach whereby our models learn from previous and incoming query executions and can, thus, ensure high efficiency and scalability as future AQ explanations do not need to access the data, a goal central within the community \cite{J13, BB07}. This is in line with recent developments in the DB community as well, such as \cite{VerdictDB}, where answers to past queries are employed to improve future query performance, although our models, like in \cite{AT15,AT17}, focuses on AQs and require no DB access at all. Further, this work centers on defining AQ explanations and deriving them accurately, scalably, and efficiently.

Scalability/efficiency is particularly important: Computing  explanations is proved to be an NP-Hard problem \cite{wang2015data} and generating them can take a long time \cite{wu2013scorpion,roy2015explaining,el2014interpretable} even with modest datasets. An exponential increase in data size 
implies a dramatic increase in the time to generate explanations. 
XAXA does not suffer from such limitations and is able to construct explanations in a matter of milliseconds even with massive data volumes. This is achieved due to two principles: 
First, on workload characteristics: workloads contain a large number of overlapping queried data subspaces, which has been acknowledged and exploited by recent research e.g., \cite{WWDI17},\cite{VerdictDB}, STRAT\cite{strat}, and SciBORQ \cite{sciborq} and has been found to hold in real-world workloads involving exploratory/statistical analysis, such as in the Sloan Digital Sky Survey and in the workload of SQLShare \cite{sqlshare}. Second, XAXA relies on a novel ML model that exploits the such workload characteristics to perform efficient and accurate (approximate) AQ explanations by monitoring query-answer pairs on-line. 

Our \textbf{contributions} are: 
\begin{enumerate}[leftmargin=*]
\item Novel explanation representation for AQs in exploratory analytics via optimal regression functions derived from a three-fold joint optimization problem.
\item A novel regression learning vector quantization statistical model for approximating the explanation functions. 
\item Model training, AQ processing and explanation generation  algorithms requiring \textit{zero} data access.
\item Comprehensive evaluation of the accuracy of AQ explanations in terms of efficiency and scalability, and sensitivity analysis.
%\item Construct novel approaches introducing hierarchical quantized multivariate piece-wise linear regression approximations of explanation functions.
%\item Extensive evaluation of the proposed explanations.
\end{enumerate}

\section{Explanation Representation}
% We formally model data items and AQs with a CRS operator as vectors in a real-valued vectorial data space, $\mathbb{D} \subset \mathbb{R}^{d}$, and vectorial query space, $\mathbb{Q} \subset \mathbb{R}^{d+1}$. 
% We chose \texttt{COUNT} and \texttt{AVG} given their popularity and their large differences which can serve as an indication of the fact that our solution can work with different types of AQs. 
%Then we formally define AQ explanations and formalize the problems and solutions for computing AQ explanations.
\label{sec:overal_idea}
%In this section we present the overall idea leading to the creation of the proposed framework. This acts as  an introduction to the concept of explanations and brief overview of the suggested ML models that are used to produce explanations for aggregate queries. We start off by describing the queries and their answers to which we seek to explain. Then move on to discuss about potential representations of the resulting explanations. Finally, we describe the process of finding an explanation for the query result and discuss different approaches. 
\subsection{Vectorial Representation of Queries}
\label{sec:queries}
Let $\mathbf{x}= [x_{1}, \ldots, x_{d}] \in \mathbb{R}^{d}$ denote a  random row vector (data point) in the $d$-dimensional data space $\mathbb{D} \subset \mathbb{R}^{d}$. 

\begin{definition}
The $p$-norm ($L_{p}$) distance between 
two vectors $\mathbf{x}$ and $\mathbf{x}'$ from $\mathbb{R}^{d}$ 
for $1 \leq p < \infty$, is $\lVert \mathbf{x} - \mathbf{x}' \rVert_{p} = (\sum_{i=1}^{d} |x_{i}-x_{i}'|^{p})^{\frac{1}{p}}$ and for $p = \infty$, is $\lVert \mathbf{x} - \mathbf{x}' \rVert_{\infty} = \max_{i=1, \ldots, d} \{ |x_{i}-x_{i}'| \}$. 
\label{definition:1}
\end{definition}

Consider a scalar $\theta >0$, hereinafter referred to as \textit{radius}, and a dataset $\mathcal{B}$ consisting of $N$ vectors $\{\mathbf{x}_{i}\}_{i=1}^{N}$.

\begin{definition}(Data Subspace)
Given $\mathbf{x} \in \mathbb{R}^{d}$
and scalar $\theta$, being the parameters of an AQ with a CRS operator, a data subspace $\mathbb{D}(\mathbf{x},\theta)$
is the convex subspace of $\mathbb{R}^{d}$,
which includes vectors $\mathbf{x}_{i}: \lVert \mathbf{x}_{i} - \mathbf{x} \rVert_{p} \leq \theta$ with $\mathbf{x}_{i} \in \mathcal{B}$. The region enclosed by a CRS operator is the referred data subspace.
\label{definition:2}
\end{definition}

% In 2-d such a subspace is defined using a range query as a selection operator for the range 
% $[x-\theta, x+\theta]$. For higher dimensions, the above subspace is defined by a hypersphere, e.g., using a UI/visulization tool whereby the user draws such hyperspheres as data spaces of interest.

\begin{definition} (Aggregate Query)
Given a data subspace $\mathbb{D}(\mathbf{x},\theta)$ an AQ $\mathbf{q} = [\mathbf{x},\theta]$ with center $\mathbf{x}$ and radius $\theta$ is represented via a regression function over $\mathbb{D}(\mathbf{x},\theta)$, that produces a response variable $y = f(\mathbf{x},\theta)$ which is the result/answer of the AQ. 
We notate with $\mathbb{Q} \subset \mathbb{R}^{d+1}$ the query vectorial space. 
\end{definition}

For instance, in the case of \texttt{AVG}, the function $f$ can be approximated by $f(\mathbf{x},\theta)= \mathbb{E}[y|\mathbf{x} \in \mathbb{D(\mathbf{x},\theta)}]$, where $y$ is the attribute we are interested in e.g., \textit{call\_time}.

% \begin{definition}(Count AQ)
%  Given a vector $\mathbf{x} \in \mathbb{R}^{d}$ and $\theta>0$, a count AQ over a dataset $\mathcal{B}$ returns the number of vectors $\lVert \mathbf{x}_{i} - \mathbf{x} \rVert_{p} \leq \theta$ where $y = n_{\theta}(\mathbf{x}) \in \{0, 1, \ldots, N\}$ is the cardinality of the set $|\{\mathbf{x}_{i} : \lVert \mathbf{x}_{i} - \mathbf{x} \rVert_{p} \leq \theta\}|$ and $\mathbf{x}_{i} \in \mathcal{B}$. We represent a count AQ as the $(d+1)$-dim. row vector $\mathbf{q} = [\mathbf{x},\theta] \in \mathbb{Q} \subset \mathbb{R}^{d+1}$. The $(d+1)$-dim. space $\mathbb{Q}$ is referred to as the query vectorial space. 
% \label{definition:3}
% \end{definition}

% \begin{definition}(Average AQ)
% Given a vector $\mathbf{x} \in \mathbb{R}^{d}$ and $\theta$, the average AQ over a dataset $\mathcal{B}$ with elements/pairs $(\mathbf{x}_{i},z_{i}) \in \mathcal{B}$ returns the average value: $z = \frac{1}{n_{\theta}(\mathbf{x})}\sum_{i \in [n_{\theta}(\mathbf{x})]}z_{i} : \lVert \mathbf{x}_{i} - \mathbf{x} \rVert_{p} \leq \theta$, where $n_{\theta}(\mathbf{x})$ is the cardinality of the set $|\{\mathbf{x}_{i} : \lVert \mathbf{x}_{i} - \mathbf{x} \rVert_{p} \leq \theta\}|$ and $(\mathbf{x}_{i},z_{i}) \in \mathcal{B}$. We represent an average AQ as the $(d+1)$-dim. row vector 
%  $\mathbf{q} = [\mathbf{x},\theta] \in \mathbb{Q} \subset \mathbb{R}^{d+1}$ and its answer as $z \in \mathbb{R}$. 
% We adopt the compact notation $i \in [n]$ as for $i =1, \ldots, n$.  
% \label{definition:4}
%  \end{definition}

\begin{definition}(Query Similarity)
The $L_{2}^{2}$ distance or similarity measure between AQs $\mathbf{q}, \mathbf{q}' \in \mathbb{Q}$
is $\lVert \mathbf{q} - \mathbf{q}' \rVert_{2}^{2} = \lVert \mathbf{x} - \mathbf{x}' \rVert_{2}^{2} + (\theta-\theta')^{2}$. 
\label{definition:5}
\end{definition}

\subsection{Functional Representation of Explanations}
The above defined AQ return a single scalar value $y = f(\mathbf{x},\theta)$ to the analysts. 
We seek to explain how such values are generated by finding a \textit{function} $f: \mathbb{R} \times \mathbb{R}^{d} \to \mathbb{R}$ that can describe how $y$, is produced given an ad-hoc query $\mathbf{q} = [\mathbf{x}, \theta]$ with $\mathbf{x} \in \mathbb{R}^{d}$. 
%A thing to note here is that the center of each query we are trying to explain is conditionally constant. This means that, 
Given our interest in parameter $\mathbf{x}$, we desire explaining the evolution/variability of output $y$ of query $\mathbf{q}$ as the radius $\theta$ varies, that is to \textit{explain} the behavior of $f$ in an area around point $\mathbf{x}$ with variable radius $\theta$. 
Therefore, given a parameter of interest $\mathbf{x}$ we define a parametric function $f(\theta; \mathbf{x})$ to which its input is the radius $\theta$. Thus, our produced explanation function(s) approximates the \textit{true} function $f$ w.r.t radius $\theta$ conditioned on the point of interest or center $\mathbf{x}$. 

An approximate explanation function can be linear or polynomial w.r.t $\theta$ given a fixed center $\mathbf{x}$. When using high-order polynomial functions for our explanations, we assume that none of the true AQ functions monotonically increases with $\theta$, 
and that we know the order of the polynomial. 
%A monotonically increasing function is a function in which the output will not decrease for an increasing input variable like $\theta$. 
A non-decreasing AQ function w.r.t $\theta$ is the \texttt{COUNT}, i.e., the number of data points in a specific (hyper)sphere of radius $\theta$. 
%Its output only increases or remains constant for an enlarging radius $\theta$ given a fixed center $\mathbf{x}$ in any data-subspace $\mathbb{D}(\mathbf{x},\theta)$. 
For aggregates such as \texttt{AVG}, the function monotonicity does not hold, thus, by representing explanation functions with high-order polynomial functions will be problematic. 
%In addition, a high-order polynomial becomes increasingly complex to interpret for an increasing order. 
By adopting only a single linear function will still have trouble representing explanations accurately. For instance, if the result $y$ is given by \texttt{AVG} or \texttt{CORR}, then the output of the function might increase or decrease for a changing $\theta$. 
Even for \texttt{COUNT}, a single linear function might be an incorrect representation as the output $y$  might remain constant within an interval $\theta$ and then increase abruptly. 
Therefore, we choose to employ multiple local linear functions to capture the possible inherent non linearity of $f$ as we vary $\theta$ around center $\mathbf{x}$ in any direction. 
%To summarize, we are faced with the following challenges: \textbf{C1} we do not adopt high-order polynomials for approximating function $f$, as they become increasingly complex to interpret and might fail when result is monotonically increasing; and/or \textbf{C2} we adopt use a single linear function for approximating function $f$ as results that oscillate cannot be accurately represented.

\begin{figure}
\begin{center}
\includegraphics[height=3.5cm,width=9.0cm]{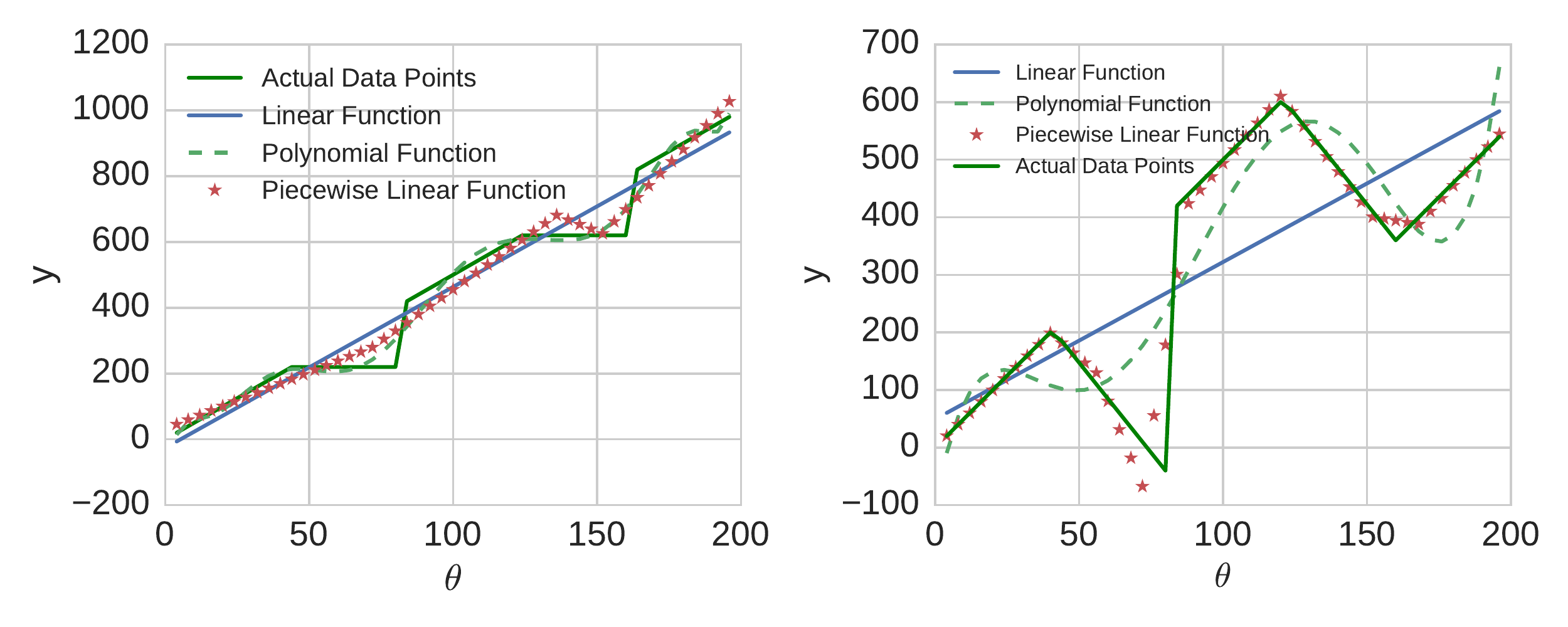}
\caption{(Left) Actual and approximate explanation functions for a monotonically increasing result (x-axis is radius $\theta$ of a query and y-axis is the result $y$ given $\mathbf{x}$); (right) Actual and approximate explanation functions for non-linear function vs. radius $\theta$ of a query given center $\mathbf{x}$.}
\label{fig:explanation_functions}
\end{center}
\end{figure}
Figure \ref{fig:explanation_functions} (left) shows that the \textit{true} function for an AQ (e.g., \texttt{COUNT}) is monotonically increasing w.r.t $\theta$ given $\mathbf{x}$. The resulting \textit{linear} and \textit{polynomial} approximation functions fail to be representative. The same holds in Figure \ref{fig:explanation_functions} (right) in which the \textit{true} function is non-linear for an AQ (e.g., \texttt{AVG}). Hence, we need to find an adjustable function that is able to capture such irregularities and conditional non-linearities. A more appropriate solution is to approximate the explanation function $f$ with a Piecewise Linear Regression (PLR) function  a.k.a. \textit{segmented regression}. 
A PLR approximation of $f$ addresses the above shortcomings by finding the \textit{best} multiple local linear regression functions, as shown in Figure \ref{fig:explanation_functions}. 
% We can envisage an \textit{explanation as the interpolating and extrapolating information derived from a given count/average query around a given center/location of interest by varying the dependent variable $\theta$ }.
%We now provide the definition of an explanation:

\begin{definition}(Explanation Function)
Given an AQ $\mathbf{q} = [\mathbf{x},\theta]$, an explanation function $f(\theta;\mathbf{x})$ is defined as the fusion of local piecewise linear regression functions $f \approx\sum{\hat{f}(\theta;\mathbf{x})}$ derived by the fitting of $\hat{f}$ over similar AQ queries $\mathbf{q}' = [\mathbf{x}',\theta']$ with radii $\theta'$ to the AQ $\mathbf{q}$.  
\label{definition:explanation}
\end{definition}
\vspace{-1pt}
\section{Explanation Approximation Fundamentals}
%\subsection{Explanation Schemes}
\label{sec:the_framework}
\subsection{Explanation Approximation}
%\label{sec:dag}
XAXA utilizes AQs to train statistical learning models that are able to accurately approximate the \textit{true} explanation function $f$ for any center of interest $\mathbf{x}$. %We consider such an approach because it can fulfill the desiderata on accuracy, scalability, and efficiency. 
Given these models, we no longer need access to the underlying data, thus, guaranteeing efficiency. This follows since the models learn the various query patterns to approximate the explanation functions $f(\theta;\mathbf{x}), \forall \mathbf{x} \in \mathbb{R}^{d}$. %In addition, exploiting insights gained and extracted from past queries assists in \textit{building} accurate approximations of explanation functions. 
%To utilize the DAG approach, previously executed queries are used to train a statistical learning model to approximate an explanation function $f$ only from past issued queries $and$ their results.

%One could utilize the available past queries and results by training a model that has the form of a PLR function. 
% Hence, for $m$ executed queries $\mathbf{q}$, we can obtain $m$ training pairs $\{(\theta_{i}, y_{i})\}, i \in [1 \ldots m]$. The difference of this vs the DG scheme is that, now, the center $\mathbf{x}_{i}$ for each query $\mathbf{q}_{i} = [\mathbf{x}_{i},\theta_{i}]$ varies, because these queries were issued at different locations/centers with different radii. 
% Hence, we train a PLR function $\hat{f}$ for approximating the actual explanation $f$ using these pairs and for each subsequent and unseen query that we seek to explain, we provide the resulting approximation model function $\hat{f}$ as our explanation. %Should I talk about how training is done ? e.g fitting a function by minimizing the loss of (y-\hat{y})^2 ?
% DAG just stores the resulting approximation model function $\hat{f}$ and, given a new AQ, the function $\hat{f}$ is used.\footnote{Note that the training queries are discarded at this point and they do not need to be stored anymore.} 

Formally, given a well defined explanation loss/discrepancy $\mathcal{L}(f,\hat{f})$ (defined later) between the \textit{true explanation function $f(\theta;\mathbf{x})$} and an \textit{approximated} function $\hat{f}(\theta;\mathbf{x})$, we seek the optimal approximation function $\hat{f}^{*}$ that minimizes the Expected Explanation Loss (EEL) for \textit{all possible} queries:
\begin{eqnarray}
\hat{f}^{*} = \arg \min_{\hat{f} \in \mathcal{F}}\int_{\mathbf{x} \in \mathbb{R}^{d}}\int_{\theta \in \mathbb{R}_{+}}\mathcal{L}(f(\theta;\mathbf{x}),\hat{f}(\theta;\mathbf{x}))p(\theta,\mathbf{x})d\theta d\mathbf{x},
\label{eq:objective}
\end{eqnarray}
where $p(\theta, \mathbf{x})$ is the probability density function of the queries $\mathbf{q} = [\mathbf{x},\theta]$ over the query workload. Eq(\ref{eq:objective}) is the objective minimization loss function given that we use the optimal model approximation function to explain the relation between $y$ and radius $\theta$ for any given/fixed center $\mathbf{x}$. 

However, accuracy will be problematic as it seems intuitively wrong that one such function can explain \textit{all} the queries at any location $\mathbf{x} \in \mathbb{R}^{d}$ with any radius $\theta \in \mathbb{R}$. 
Such an explanation is not accurate because: (i) the radius $\theta$ can generally be different for different queries. For instance, crime data analysts will issue AQs with a different radius to compare if crimes centered within a small radius at a given location increase at a larger/smaller radius; (ii) At different locations $\mathbf{x}$, the result $y$ will almost surely be different. 
% \begin{figure}
% \begin{tabular}{cc}
% \includegraphics[height=3.5cm,width=8.0cm]{global_problem}
% \end{tabular}
% \caption{One global PLR function $y = f(\theta;\mathbf{x})$ for all possible queries yields small accuracy. The three distinct clusters represent multiple queries with a given radius and result. Local Models focus on smaller local spaces yielding more accurate explanations for related queries.}
% \label{fig:global_local}
% \end{figure}
% \begin{figure}
% \begin{center}
% 
% \caption{One global PLR function $y = f(\theta;\mathbf{x})$ for all possible queries yields small accuracy. The three distinct clusters represent multiple queries with a given radius and result.}
% \label{fig:global_prob}
% \end{center}
% \end{figure}

%\subsection{Local Explanation Approximation}
Therefore, we introduce \textit{local} approximation functions $\hat{f}_1,\ldots, \hat{f}_K$ that collectively minimize the objective in (\ref{eq:objective}). Thus, we no longer wish to find one global approximation to the true explanation function for all possible queries. Instead, we fit a number of local optimal functions that can explain a \textit{subset} of possible queries. We refer to those models as Local PLR Models (LPM).  

The LPMs are fitted using AQs forming a cluster of similar query parameters. For each uncovered cluster of AQs, we fit an LPM. 
Therefore, if $K$ clusters are obtained from clustering the query vectorial space $\mathbb{Q}$, then $K$ LPMs are fitted.
The obtained $K$ LPMs are essentially local multiple possible explanations. Intuitively, this approach makes more sense as queries that are issued at relatively close proximity to each other, location-wise, and that have similar radii will tend to have similar results/answers and in turn will behave more or less the same way. The resulting fused explanation has more accuracy as the clusters have less variance for $y$ and $\theta$. This was empirically shown from our experimental workload.
% \begin{table}
% \centering
% \caption{Variance for $\theta$ and result $y$(or $z$) is reduced by the use of Local models.(Results from experimental workload described in Experimental Evaluation)}
% \begin{tabular}{ |c|c|c| } 
%  \hline
%  & \textbf{Global Model $\hat{f}$} & \textbf{Local Models $\hat{f}_{k}$} \\ \hline
%  \textbf{Variance $\theta$} & $0.124$ &  $0.0017$ \\ \hline
%  \textbf{Variance $y$} & $1.5\cdot10^{10}$ & $3.5\cdot10^8$ \\ \hline
% \end{tabular}

% \label{tab:variance}
% \end{table}
% \begin{figure}
% \begin{center}
% \includegraphics[height=3.5cm,width=8.0cm]{local_models}
% \caption{Local Models focus on smaller local spaces yielding more accurate explanations for related queries.}
% \label{fig:local_models}
% \end{center}
% \end{figure}

Formally, to minimize EEL, we seek $K$ local approximation functions $\hat{f}_{k} \in \mathcal{F}, k=1 \in [K]$ from a family of linear regression functions $\mathcal{F}$, such that for each query $\mathbf{q}$ belonging to the partition/cluster $k$ of the query space, notated as $\mathbb{Q}_{k}$, the summation of the local EELs is minimized:
\begin{eqnarray}
\mathcal{J}_{0}(\{\hat{f}_{k}\}) = \sum_{\hat{f}_{k} \in \mathcal{F}}\int_{\mathbf{q} \in \mathbb{Q}_{k} \subset \mathbb{R}^{d+1}}\mathcal{L}(f(\theta;\mathbf{x}),\hat{f}_{k}(\theta;\mathbf{x}))p_{k}(\mathbf{q})d\mathbf{q}
\label{eq:lobjective}
\end{eqnarray}
where $p_{k}(\mathbf{q})$ is the probability density function of the query vectors belonging to a query subspace $\mathbb{Q}_{k}$. Thus, $\mathcal{J}_{0}$ forms our \textit{generic} optimization problem mentioned in Section \ref{sec:prob_def}.
%For example, crime data analysts tasked with investigating certain regions, issue queries at different locations $\mathbf{x}$, with varying $\theta$. Therefore, the log containing all the queries that each one of them has issued will be diverse. It then makes more sense to find the specific regions / subspaces $\mathbb{Q}_{k}$ that they targeted, and construct explanations over those regions that contain similar queries and, thus, similar results. This approach is guaranteed to be more accurate.
%\subsection{Measuring Expected Explanation Loss}
We define \textit{explanation loss} $\mathcal{L}(f,\hat{f})$ as the discrepancy of the actual explanation function $f$ due to the 
approximation of explanation $\hat{f}$. 
For evaluating the loss $\mathcal{L}$, we propose two different aspects: \textit{(1) the statistical aspect}, where the goodness of fit of our explanation function is measured and (2) the \textit{predictive accuracy} denoting how well the results from the \textit{true} explanation function can be approximated using our explanation function; see Section \ref{sec:evaluation}.

\subsection{Solution Overview}
% The previous section was an overview of the overall idea behind explaining aggregates, how an explanation can be represented and how one can go about generating explanations. We saw both Data-Agnostic and Data-Agnostic approaches along with their advantages and disadvantages. We now move on to dive into the details of the implemented framework. The next sections, will give the reader a brief overview of the whole framework and then each section will be devoted to a central part of the system.
%\subsection{General Overview}
%Our goal is an approximation explanation framework that provides an explanation to the analyst of the results $y$ and $z$ of AQs. Explanations are represented as a fusion of explanation functions $f$ minimizing the global and local objectives in (\ref{eq:objective}) and (\ref{eq:lobjective}), respectively. %To that extent we can adopt a DG or a DaG approach to find such a function. 
%We introduce both schemes DG and (L)DaG in our solution and give a detailed comparison using a variety of metrics for evaluating the EEL in Section \ref{sec:evaluation}. As mentioned, LDAG, by building LPMs, improves accuracy. LDAG exploits past queries and their results for training a statistical learning model consisting of $K$ LPMs.
The proposed methodology for computing explanations is split into three \textit{modes} (Figure \ref{fig:modes}). 
\begin{figure}
\begin{center}
\includegraphics[height=2cm,width=9.0cm]{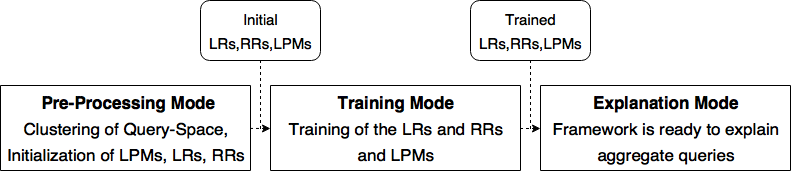}
\caption{The XAXA framework modes overview. }
\label{fig:modes}
\end{center}
\end{figure}
The \textit{Pre-Processing Mode} aims to identify the optimal number of LPMs and an initial approximation of their parameters using previously executed queries. The purpose of this mode is basically to jump-start our framework. In the \textit{Training Mode}, the LPMs' parameters are incrementally optimized to minimize the objective function (\ref{eq:lobjective}) as incoming queries are processed in an on-line manner. In the \textit{Explanation Mode}, the framework is ready to explain AQ results via the provided local explanation functions.

\subsubsection{Pre-Processing Mode}
A training set $\mathcal{T}$ of $m$ previously executed queries $\mathbf{q}$ and their corresponding aggregate results is used as input to the \textit{Pre-Processing Mode}.
%e.g., count $y$: $\mathcal{T} = \{(\mathbf{q}_{1},y_{1}), \ldots, (\mathbf{q}_{m},y_{m})\}$, where $\mathbf{q}=(\mathbf{x}, \theta)$. 
% We similarly obtain the training set for average queries, i.e., their results are the average $z$ values. 
The central task is to partition (quantize) the query space $\mathbb{Q}$ based on the observed previous queries $\mathbf{q} \in \mathcal{T}$ into $K$ clusters, sub-spaces $\mathbb{Q}_{k}$, in which queries with similar $\mathbf{x}$ are grouped together. Each cluster is then further quantized into $L$ sub-clusters, as queries with similar $\mathbf{x}$ are separated with regards to their $\theta$ parameter values. Therefore, this is a hierarchical query space quantization (first level partition w.r.t. $\mathbf{x}$ and second level partition w.r.t. $\theta$), where each Level-1 (L1) cluster $\mathbb{Q}_{k}, k = 1,\ldots,K$ is associated with a number of Level-2 (L2) sub-clusters $\mathbb{U}_{kl}, l = 1, \ldots, L$ in the $\theta$ space. 
For each L1 cluster $\mathbb{Q}_{k}$ and L2 sub-cluster $\mathbb{U}_{kl}$, 
we assign a L1 representative, hereinafter referred to as Location Representative (LR) and a L2 representative, hereinafter referred to as Radius Representative (RR). 
%in the center/location space and in the radius space, respectively. 
The LR converges to the mean vector of the centers of all queries belonging to L1 cluster $\mathbb{Q}_{k}$, 
while the associated RR converges to the mean radius value of the radii of all queries, whose radius values belong to $\mathbb{U}_{kl}$.
After the hierarchical quantization of the query space, the task is to \textit{associate an LPM} $\hat{f}_{kl}(\theta;\mathbf{x})$ with \textit{each} L2 sub-cluster $\mathbb{U}_{kl}$, given that the center parameter $\mathbf{x}$ is a member of the corresponding L1 cluster $\mathbb{Q}_{k}$. This process is nicely summarized in Figure \ref{fig:modelSummary}.

\subsubsection{Training Mode}
This mode optimally adapts/tunes the clustering parameters (LR and RR representatives) of the pre-processing mode in order to minimize the objective function (\ref{eq:lobjective}). This optimization process is achieved incrementally by processing each new pair $(\mathbf{q}_{i},y_{i})$ in an on-line manner. 
Consulting Figure \ref{fig:modelSummary}, in \textit{Training} mode, each incoming query $\mathbf{q}_{i}$ is projected/mapped to the closest LR corresponding to a L1 cluster. 
Since, the closest LR is associated with a number of L2 RRs, the query is then assigned to one of those RRs (closest $u$ to $\theta$), and then the associated representatives are adapted/fine-tuned. After a pre-specified number of processed queries, the corresponding LPM $\hat{f}_{kl}(\theta;\mathbf{x})$ is re-adjusted to associate the result $y$ with the $\theta$ values in $\mathbb{U}_{kl}$ and account for the newly associated queries.

\subsubsection{Explanation Mode}
In this mode no more modifications to LRs, RRs, and the associated LPMs are made. Based on the L1/2 representatives and their associated approximation explanation functions LPMs, the model is now ready to provide explanations for AQs. 
Figure \ref{fig:modelSummary} sums up the result of all three modes and how an explanation is given to the user. For a given query $\mathbf{q}$, the XAXA finds the closest LR;$\mathbf{w}_k$ and then, based on a combination of the RRs;$(u_{k,1},\ldots, u_{k,3})$ and their associated LPMs;$(\hat{f}_{k,1}\ldots,\hat{f}_{k,3})$, returns an explanation as a fusion of diverse PLR functions derived by the L2 level. We elaborate on this fusion of L2 LPMs in Section \ref{ref:prediction_mode}\footnote{Note that LPMs are also referred to as PLRs.}.  
\section{Optimization Problems Deconstruction}
%\subsection{Optimization Problems}
%\subsubsection{Explanation Representatives}
% We initially defined the generic optimization problem in (\ref{eq:lobjective}), which identifies the need for local approximations of the true explanation function. Subsequently, we defined a high level solution that allows us to tackle such problem. In this section, we de-construct the generic problem into more specific optimization problems, where XAXA minimizes the EEL.5

\subsection{Optimization Problem 1: Query Space Clustering}
The first part of the deconstructed generic problem identifies the need to find \textit{optimal} LRs and RRs, as such optimal parameters guarantee better grouping of queries thus better approximation of \textit{true} function, during the \textit{Pre-Processing} phase. The LRs are initially random location vectors $\mathbf{w}_{k} \in \mathbb{R}^{d}, k = 1, \ldots, K$, and are then refined by a clustering algorithm as the mean vectors of the query locations that are assigned to each LR.  Formally, this phase finds the optimal mean vectors $\mathcal{W} = \{\mathbf{w}_{k}\}_{k=1}^{K}$, which minimize the L1 Expected Quantization Error (L1-EQE):
\begin{eqnarray}
\mathcal{J}_{1}(\{\mathbf{w}_{k}\}) = \mathbb{E}[\lVert \mathbf{x}-\mathbf{w^*}\rVert^{2}; \mathbf{w}^{*} = \arg \min_{k=1,\ldots,K} \lVert \mathbf{x}-\mathbf{w}_{k}\rVert^{2}],
\label{eq:kmeans-association}
\end{eqnarray}
where $\mathbf{x}$ is the location of query $\mathbf{q} = [\mathbf{x},\theta] \in \mathcal{T}$ and $\mathbf{w}_k$ is the mean center vector of all queries $\mathbf{q} \in \mathbb{Q}_{k}$ associated with $\mathbf{w}_k$. 
We adopt the $K$-Means \cite{hartigan1979algorithm} clustering algorithm to identify the L1 LRs based on the queries' centers $\mathbf{x}$[\footnote{The framework can employ any clustering algorithm.}]. This phase yields the $K$ L1 cluster representatives, LRs.
A limitation of the $K$-Means algorithm is that we need to specify $K$ number of LRs in advance. 
Therefore, we devised a simple strategy to find a near-optimal number $K$. By running the clustering algorithm in a loop, each time increasing the input parameter $K$ for the $K$-Means algorithm, we are able to find a $K$ that is near-optimal. In this case, an optimal $K$ would \textit{sufficiently} minimize the Sum of Squared Quantization Errors (SSQE), which is equal to the summation of distances, of all queries from their respective LRs. The algorithm is shown in Appendix \ref{algo:near-optimal}.
\begin{equation}
\text{SSQE} = \sum_i^n \min_{\mathbf{w}_k\in \mathcal{W}}(||\mathbf{x}_i - \mathbf{w}_k ||_2^2)
\label{eq:ssqe}
\end{equation}
This algorithm starts with an initial $K$ value and gradually increments it until SSQE yields an improvement no more than a predefined threshold $\epsilon > 0$. %The result of the clustering algorithm is the set of L1 LRs $\mathcal{W} = \{\mathbf{w}_{k}\}_{k=1}^{K}$, which minimize (\ref{eq:kmeans-association}).

We utilize $K$-Means over each L1 cluster of queries created by the L1 query quantization phase. Formally, we minimize the conditional L2 Expected Quantization Error (L2-EQE):
\begin{eqnarray}
\mathcal{J}_{1.1}({u_{k,l}}) = \mathbb{E}[(\theta-u_{k,*})^2]; u_{k,*} = \arg \min_{l\in L}{(\theta - u_{k,l})^2}
\label{eq:rr-kmeans-association}
\end{eqnarray}
conditioned on $\mathbf{w}_k = \mathbf{w}^*$. Therefore, for each LR $\mathbf{w}_{1},\ldots ,\mathbf{w}_{K}$, we \textit{locally} run the $K$-Means algorithm with $L$ number of RRs, where a near optimal value for $L$ is obtained following the same near-optimal strategy. With SSQE being computed using $u$ and $\theta$ instead of ($\mathbf{w}, \mathbf{x})$. Specifically, 
we identify the L2 RRs over the radii of those queries from $\mathcal{T}$ whose closest LR is $\mathbf{w}_{i}$. 
Then, by executing the $L$-Means over the radius values 
from those queries we derive the corresponding set of radius representatives $\mathcal{U}_{i} = \{u_{i1},\dots ,u_{iL}\}$, where each $u_{il}$ is the mean radius of all radii in the $l$-th L2 sub-cluster of the $i$-th L1 cluster. Thus the first part of the deconstructed optimization problem can be considered as two-fold, as we wish to find optimal parameters for both LRs and RRs that minimize (\ref{eq:kmeans-association}) and (\ref{eq:rr-kmeans-association}).

\subsection{Optimization Problem 2: Fitting PLRs per Query Cluster}
%\textbf{Third Phase: PLR Initial Fitting}.
The second part of the deconstructed generic optimization problem has to do with fitting \textit{optimal} approximation functions such that the local EEL is minimized given the optimal parameters obtained from the first part of the deconstructed problem.
We fit (approximate) $L$ PLR functions $\hat{f}_{kl}(\theta;\mathbf{w}_{k})$ for each L2 sub-cluster $\mathbb{U}_{kl}$ for those $\theta$ values that belong to the L2 RR $u_{kl}$. The fitted PLR captures the \textit{local} statistical dependency 
of $\theta$ around its mean radius $u_{kl}$ given that $\mathbf{x}$ is a  member of the L1 cluster represented by $\mathbf{w}_{k}$. Given the objective in (\ref{eq:lobjective}), for each local L2 sub-cluster, the approximate function $\hat{f}_{kl}$ minimizes the conditional Local EEL: 
\begin{eqnarray}\label{eq:llobjective}
\mathcal{J}_{2}(\{\beta_{kl},\lambda_{kl}\}) & = & \mathbb{E}_{\theta,\mathbf{x}}[\mathcal{L}(f_{kl}(\theta;\mathbf{w}_{k}),\hat{f}_{kl}(\theta;\mathbf{w}_{k}))] \\ \nonumber
\mbox{s.t.} & & \mathbf{w}_{k} = \arg \min_{j \in [K]}\lVert \mathbf{x} - \mathbf{w}_{j}\rVert^{2}, \\  \nonumber
& & u_{kl} = \arg \min_{j \in [L]} | \theta - u_{kl} |
\end{eqnarray}
conditioned on the closeness of the query's $\mathbf{x}$ and $\theta$ to the L1 and L2 quantized query space $\mathbb{Q}_{k}$ and $\mathbb{U}_{kl}$, respectively, where $\{\beta_{kl},\lambda_{kl}\}$ are the parameters for the PLR local approximation explanation functions $\hat{f}_{kl}$ defined as follows in (\ref{eq:mars}). 

\textbf{Remark 1:} Minimizing objective $\mathcal{J}_{2}$ in (\ref{eq:llobjective}) is not trivial due to the double conditional expectation over each query center and radius. To \textit{initially} minimize this local objective, we adopt the Multivariate Adaptive Regression Splines (MARS) \cite{friedman1991multivariate} as the approximate model explanation function $\hat{f}_{kl}$. %because of its ability to capture nonlinearities in the $\theta$ space conditioned around a given location vector $\mathbf{w}_{k}$ by introducing hinge points. 
Thus, our approximate $\hat{f}_{kl}$ has the following form: 
\begin{eqnarray}
\hat{f}_{kl}(\theta;\mathbf{w}_{k}) & = & \beta_{0} + \sum_{i=1}^{M}\beta_{i}h_{i}(\theta),
\label{eq:mars}
\end{eqnarray}
where the basis function $h_i(\theta) = \max\{0, \theta-\lambda_{i}\}$, 
with $\lambda_{i}$ being the \textit{hinge} points. 
Essentially, this creates $M$ linear regression functions.
% $h_{i}$, where the hinge points $\lambda_{i}$ and the constants $\beta_{0}, \ldots, \beta_{M}$ determine the points over which
% each function is used. 
The number $M$ of linear regression functions is automatically derived by MARS using a threshold for convergence w.r.t $R^2$(\textit{coefficient-of-determination}, later defined); which optimize fitting. Thus, guaranteeing an optimal number of $M$ linear regression functions.  For each L2 sub-cluster, we fit $L$ MARS functions $\hat{f}_{kl}, l = 1, \ldots, L$, each one associated with an LR and an RR, thus, in this phase we initially fit $K \times L$ MARS functions for providing explanations for the whole query space.    
Figure \ref{fig:modelSummary} illustrates the two levels L1 and L2 of our explanation methodology, where each LR and RR are associated with a MARS model. 

\subsection{Optimization Problem 3: Putting it All Together}
The optimization objective functions in (\ref{eq:kmeans-association}), (\ref{eq:rr-kmeans-association}) and (\ref{eq:llobjective}) should be combined to establish our generic optimization function $\mathcal{J}_{0}$, which involves the estimation of the \textit{optimal} parameters that minimize the EQE in $\mathcal{J}_{1}$ (L1) and $\mathcal{J}_{1.1}$ (L2), and then the conditional optimization of the parameters in $\mathcal{J}_{2}$. In this context, we need to estimate these values of the parameters in $\mathcal{W} = \{\mathbf{w}_{k}\}$ and 
$\mathcal{U} = \{u_{kl}\}$ that minimize the EEL given that our explanation comprises a set of regression functions $\hat{f}_{kl}$. Our joint optimization is optimizing \textit{all} the parameters from $\mathcal{J}_{1}$, $\mathcal{J}_{1.1}$, and $\mathcal{J}_{2}$ from Problems 1 and 2:

\begin{eqnarray}
\mathcal{J}_{3}(\mathcal{W}, \mathcal{U}, \mathcal{M}) = \mathcal{J}_{1}(\mathcal{W}) + \mathcal{J}_{1.1}(\mathcal{U}) + \mathcal{J}_{2}(\mathcal{M})
\label{eq:J3}
\end{eqnarray}
with parameters: 

\begin{eqnarray}
\mathcal{W} = \{\mathbf{w}_{k}\}, \mathcal{U} = \{u_{kl}\}, \mathcal{M} =  \{(\beta_{i},\lambda_{i})_{kl}\}
\label{eq:parameters}
\end{eqnarray}
with $k \in [K]$, $l \in [L]$, $i \in [M]$, which are stored for fine tuning and explanation/prediction.

\textbf{Remark 2:} The optimization function $\mathcal{J}_{3}$ approximates the generic objective function $\mathcal{J}_{0}$ in (\ref{eq:lobjective}) via L1 and L2 query quantization (referring to the integral part of (\ref{eq:lobjective})) and via the estimation of the local PLR functions referring to the family of function space $\mathcal{F}$. Hence, we hereinafter contribute to an algorithmic/computing solution to the optimization function $\mathcal{J}_{3}$ approximating the \textit{theoretical} objective function $\mathcal{J}_{0}$. 

\begin{figure}
\begin{center}
\includegraphics[height=4.5cm,width=8.5cm]{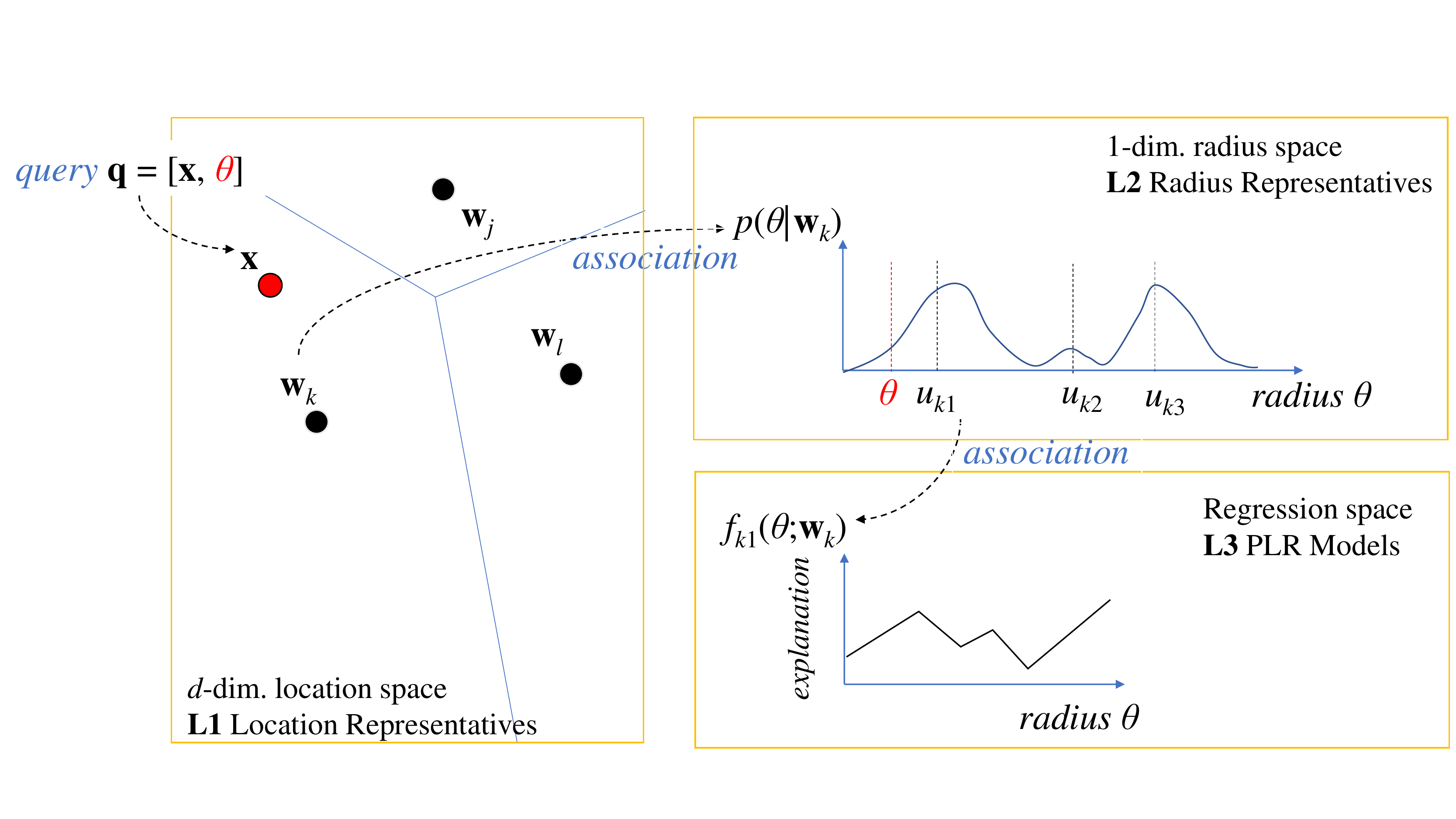}
%{modelSummary}
\caption{The XAXA framework rationale: Query mapping to L1 cluster, then mapping to L2 sub-cluster, and association to L3 PLR regression space. The explanation is provided by the associated set of PLR functions $f_{kl}$.}
\label{fig:modelSummary}
\end{center}
\end{figure}

\section{XAXA: Statistical Learning Methodology}
We propose a new statistical learning model that associates the (hierarchically) clustered query space with (PLR-based explanation) functions.Given the hierarchical query space quantization and local PLR fitting, the training mode fine-tunes the parameters in (\ref{eq:parameters}) to optimize both $\mathcal{J}_{1}$, $\mathcal{J}_{1.1}$ in (\ref{eq:kmeans-association}), (\ref{eq:rr-kmeans-association}) and $\mathcal{J}_{2}$ in (\ref{eq:llobjective}). 
The three main sets of parameters $\mathcal{W}$, $\mathcal{U}$, and $\mathcal{M}$ of the framework are \textit{incrementally} trained in \textit{parallel} using queries issued against the DBMS.
Therefore, (i) the analyst issues an AQ $\mathbf{q} = [\mathbf{x}, \theta]$; (ii) the DBMS answers with result $y$ ; (iii) our framework exploits pairs $(\mathbf{q},y)$ to train its new statistical learning model. 

% \begin{figure}
% \begin{center}
% \includegraphics[height=3cm,width=7.0cm]{training_mode_parallel}
% \caption{Framework Training: Observe issued queries and their answers to online train and adjust the explanation models.}
% \label{fig:training_parallel}
% \end{center}
% \end{figure}

Training follows three steps. First, in the \textit{assignment} step, the incoming query $\mathbf{q}$ is associated with an L1, L2 and PLR, by finding the closest representatives at each level. In the \textit{on-line adjustment} step, the 
LR, RR and PLR representative parameters are gradually modified 
w.r.t. the associated incoming query in the direction of minimizing the said objectives. Finally, the \textit{off-line adjustment} step conditionally fine-tunes any PLR fitting model associated with any incoming query so far. This step is triggered when a predefined retraining value is reached or the number of queries exceeds a threshold. 

\textbf{Query Assignment Step}. For each executed query-answer pair $(\mathbf{q},y)$, we project $\mathbf{q}$ to its closest L1 LR using only the query center $\mathbf{x}$ based on (\ref{eq:kmeans-association}). Obtaining the projection $\mathbf{w}_{k}^{*}$  
allows us to \textit{directly} retrieve the associated RRs $\mathcal{U}_{k}^{*} = \{u_{1k}^{*}, \ldots, u_{Lk}^{*}\}$. 
Finding the best L2 RR in $\mathcal{U}_{k}^{*}$ is, however, more complex than locating the best LR. 
In choosing one of the L2 RRs, we consider both \textit{distance} and the associated \textit{prediction error}. 
Specifically, the \textit{prediction error} is obtained by the PLR fitting models of each RR from the set $\mathcal{U}_{k}^{*}$. 
Hence, in this context, we first need to consider the distance of our query $\mathbf{q} = [\mathbf{x},\theta]$ to all of the RRs in $\mathcal{U}_{k}^{*}$ focused on the absolute  distance in the radius space: 
\begin{equation}
|\theta-u_{kl}|, \forall u_{kl} \in \mathcal{U}_{k}^{*},
\label{eq:dist-error}
\end{equation}
and, also, the \textit{prediction error} given by each RR's associated PLR fitting model $\hat{f}_{kl}$. The prediction error is obtained by the squared difference of the actual result $y$ and the predicted outcome of the local PLR fitting model $\hat{y} = \hat{f}_{kl}(\theta;\mathbf{w}_{k}^{*})$:
\begin{equation}
(y - \hat{f}_{kl}(\theta;\mathbf{w}_{k}^{*}))^2, l =1, \ldots, L 
\label{eq:pred-error}
\end{equation}

Therefore, for assigning a query $\mathbf{q}$ to a L2 RR, we combine both distances in (\ref{eq:dist-error}) and (\ref{eq:pred-error}) to get the assignment distance in (\ref{eq:distance-error}), which returns the RR in $\mathcal{U}_{k}^{*}$ which minimizes:

\begin{equation}
  l^* = \arg \min_{l \in [L]}(z |\theta-u_{kl}| + (1-z)(y - \hat{f}_{kl}(\theta;\mathbf{w}_{k}^{*}))^2)
 \label{eq:distance-error}
\end{equation}
The parameter $z \in (0,1)$ tilts our decision towards the \textit{distance}-wise metric or the \textit{prediction}-wise metric, depending on which aspect we wish to attach greater significance.
%. For instance, if we prefer to assign queries to RRs that have a smaller prediction error, we decrease the value of $\mu$, and increase $\mu$ if we desire our decision to be based only on the distance in radius space.

% \textbf{Remark 3:} \textit{Why incorporate prediction error?} 
% We could associate an incoming query with the closest L2 RR as is being done with L1 LR. However, 
% note that an explanation function may have lower prediction error even though is not the closest (w.r.t to RR). Intuitively, this holds true, as some function might be able to make better generalizations even if their RRs are further apart. Therefore, we introduce the weighted-distance in (\ref{eq:distance-error}) to account for this and make more sophisticated selections.

\textbf{On-line Representatives Adjustment Step.} This step optimally adjusts the positions of the chosen LR and RR so that training is informed by the new query. Their positions are shifted using Stochastic Gradient Descent (SGD) \cite{bottou2012stochastic} over the $\mathcal{J}_{1}$ and $\mathcal{J}_{2}$ w.r.t. $\mathbf{w}$ and $\theta$ variables in the negative direction of their gradients, respectively. This ensures the \textit{optimization of both objective functions}. Theorems \ref{theorem:1} and \ref{theorem:2} present the update rule for the RR selected in (\ref{eq:distance-error}) to minimize the EEL given that a query is projected to its L1 LR and its convergence to the median value of the radii of those queries.

\begin{theorem}
Given a query $\mathbf{q} = [\mathbf{x},\theta]$ projected onto the closest L1 $\mathbf{w}_{k^{*}}$ and L2 $u_{k^{*},l^{*}}$, the update rule for $u_{k^{*},l^{*}}$ that minimizes $\mathcal{J}_{2}$ is:  
\begin{equation}
\Delta u_{k^{*},l^{*}} \leftarrow \alpha z \mbox{sgn}(\theta - u_{k^{*},l^{*}})
\label{eq:sgd}
\end{equation}
\label{theorem:1}
\end{theorem}

\begin{proof}
Proof is in the Appendix \ref{appendix:theorem:1}
\end{proof}

$\alpha \in (0,1)$ is the learning rate defining the shift of $\theta$ ensuring convergence to optimal position and $sgn(x) = \frac{d|x|}{dx}, x \neq 0$ is the signum function. Given that query $\mathbf{q}$ is projected on L1 $\mathbf{w}_{k}^{*}$ and on L2 $u_{k^{*},l^{*}}$, the corresponding RR converges to the local median of all radius values of those queries. 

\begin{theorem}
Given the optimal update rule in (\ref{eq:sgd}) for L2 RR $u_{k,l}$, it converges to the median of the $\theta$ values of those queries projected onto the L1 query subspace $\mathbb{Q}_{k}$ and the L2 sub-cluster $\mathbb{U}_{kl}$, i.e., for each query $\mathbf{q} = [\mathbf{x},\theta]$ with $\mathbf{x} \in \mathbb{Q}_{k}$, it holds true for $u_{kl} \in \mathbb{U}_{kl}$ that: $\int_{0}^{u_{kl}}p(\theta|\mathbf{w}_{k})d\theta = \frac{1}{2}$.
\label{theorem:2}
\end{theorem}

\begin{proof}
Proof is in the Appendix \ref{appendix:theorem:2}
\end{proof}
Using SGD, $\theta_{k^{*},l^{*}}$ converges to the median of all radius values of all queries in the local L2 sub-cluster in an on-line manner. %The alternative would have been to calculate the median of all training queries associated so far, every time a new query was associated. Using SGD, we drastically cut down our training time for convergence to the optimal positions of L2 RRs.

\textbf{Off-line PLR Adjustment Step}. The mini-batch adjustment step is used to conditionally re-train the PLRs to reflect the changes by (\ref{eq:sgd}) in $\mathcal{U}_{k}$ parameters. As witnessed earlier, representatives are incrementally adjusted based on the projection of the incoming query-answer pair onto L1 and L2 levels. For the PLR functions, the adjustment of hinge points and parameters $(\beta_{i},\lambda_{i})$ needs to happen in mini-batch mode taking into consideration the projected incoming queries onto the L2 level.
To achieve this, we keep track of the number of projected queries on each L2 sub-cluster $\mathcal{U}_{k}$ and re-train the corresponding parameters $(\beta_{kli},\lambda_{kli})$ PLR of the fitting $\hat{f}_{kl}$ given a conditionally optimal L2 RR $u_{kl}$ and for every processed query we increment a counter. Once we reach a predefined number of projected queries-answers, we re-train every PLR model that was 
affected by projected training pairs.
% Once the fitting models are retrained, then 
% a new  era of on-line adjustment begins 
% until the end of the training pairs or the convergence of the L2 RRs.  
%%Why Training Mode and not just Pre-Processing Mode
% \paragraph{Why Training Mode and not just Pre-Processing Mode}
% Although the differences between \textit{Pre-Processing Mode} and \textit{Training Mode} are obvious, there has not been a solid explanation as to why these two need to coexist instead of having just one of them to do the job for both. For instance, we could initialize  The task of initializing LRs, RRs and MARS models cannot happen at the Training stage because queries are processed online and the assignment of queries to RRs takes into consideration the prediction error which cannot be obtained without the initialized MARS models. 
\vspace{-20pt}
\section{XAXA: Explanation Serving}
\label{ref:prediction_mode}
After \textit{Pre-Processing} and \textit{Training}, explanations can be provided for \textit{unseen} AQs.
%No modification of the model parameters are performed and explanations do not require executing any AQs. 
The explanations are \textit{predicted} by the fusion of the trained/fitted PLRs based on the incoming queries. The process is as follows. The analyst issues a query $\mathbf{q}$, $q \notin \mathcal{T}$. 
%XAQ replies with an explanation function which will be used by the analyst for subsequent analysis. 
To obtain the explanation of the result for the given query $\mathbf{q} = [\mathbf{x},\theta]$, XAXA finds the closest $\mathbf{w}_{k}^{*}$ LR 
and then directly locates all of the associated RRs of the $\mathcal{U}_{k}^{*}$. The provided explanation utilizes all \textit{selected} associated $L$ PLRs fitting models $\hat{f}_{k^{*},l}$. The selection of the most relevant PLR functions for explaining $\mathbf{q}$ is based on the boolean indicator $I(\theta, u_{k^{*},l})$, where $I(\theta, u_{k^{*},l}) = 1$ if we select to explain AQ based on the area around $\theta$ represented by the L2 RR $u_{k^{*},l}$; 0 otherwise. Specifically this indicator is used to denote which is the domain value for the $\theta$ radius in order to deliver the dependency of the answer $y$ within this domain as reflected by the corresponding PLR $\hat{f}_{k^{*},l}(\theta; \mathbf{w}_{k^{*}})$. In other words, for different query radius values, XAXA  selects different PLR models to represent the AQ, which is associated with the RR that is closer to the given query's $\theta$. Using this method, $L$ possible explanations for $\mathbf{q}$ exist and the selection of the most relevant PLR fitting model for the current radius domain is determined when $I(\theta, u_{k^{*},l}) = 1$ and, simultaneously, when:
\begin{equation}
\label{eq:boolean_function}
\begin{cases}
0\leq \theta<u_{k^{*},1}+\frac{1}{2}|u_{k^{*},1}-u_{k^{*},2}| \text{ and } l = 1\\
u_{k^{*},l-1} + \frac{1}{2}|u_{k^{*},l-1}+u_{k^{*},l}| \leq \theta < u_{k^{*},l} + \frac{1}{2}|u_{k^{*},l}+ u_{k^{*},l+1}|\\ \qquad \qquad \qquad \qquad \qquad \text{ and } l < L\\
u_{k^{*},l-1}+\frac{1}{2}|u_{k^{*},l-1}+u_{k^{*},l}|\leq \theta \text{ and } l = L
\end{cases}
\end{equation}

{\bf Interpretation}. The domain radius in (\ref{eq:boolean_function}) shows the radius interval boundaries for selecting the most relevant PLR fitting model for explanation in the corresponding radius domain. 
In other words, we switch between PLR models according to the indicator function. For instance, from radius $0$ up to the point where the first RR would still be the closest representative, indicated as $\leq u_1+\frac{1}{2}|u_1-u_2|$, we use the first PLR fitting model. Using this selection method, we derive an accurate explanation for the AQ represented by:
\begin{equation}
\label{eq:summation-of-f}
\sum_{l=1}^LI(\theta, u_{k^{*},l})\hat{f}_{k^{*},l}
\end{equation}
which is a PLR function where the indicator $I(\cdot, \cdot)$ returns 1 only for the selected PLR.
\textbf{Remark 2.} It might seem counterintuitive to use a combination of functions instead of just one PLR model to explain AQs. However, accuracy is increased when such alternative explanations are given as we are better able to explain the query answer $y$ and/or $z$ for a changing $\theta$ with the fitting model that is associated with the closest L2 RR. This is illustrated in Figure \ref{fig:explain-better} where a combination (fusion) of PLR models is able to give a finer grained explanation than a single model. In addition, we can visualize the intervals where each one of those models explain the dependency of the query answer w.r.t. a change in the radius.

\begin{figure}[!htbp]
\begin{center}
\includegraphics[height=3.5cm,width=7.0cm]{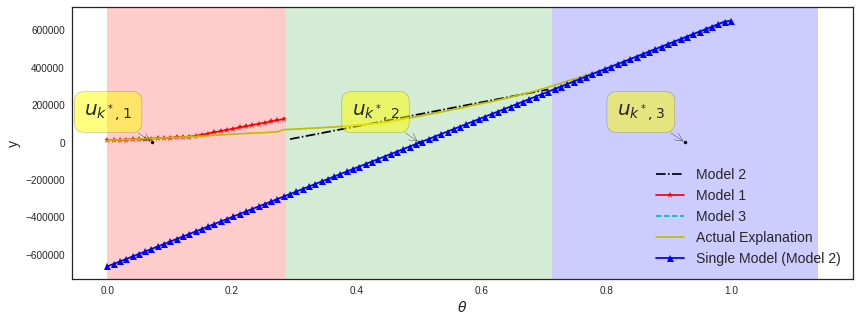}
\caption{Better explanations with $>1$ relevant models are obtained as evidenced by the intervals (radius domains), where each model is selected along with their associated L2 RRs.}
\label{fig:explain-better}
\end{center}
\end{figure}
% \vspace{-1pt}
\section{Experimental Evaluation}
\label{sec:evaluation}
\subsection{Experimental Setup \& Metrics}
\textbf{Data Sets and Query Workloads:} The real dataset  $\mathcal{B}_1 = \{\mathbf{x}_{i}\}_{i=1}^{N}, \mathbf{x} \in \mathbb{R}^2$ with cardinality $|\mathcal{B}_1| = N = 6\cdot{10^6}$ contains crimes for the city of Chicago \cite{crimesdata}. We also used another real dataset $\mathcal{B}_2= \{\mathbf{x}\}$ where $\mathbf{x}\in\mathbb{R}^2$ and  $|\mathcal{B}_2| = 5.6 \cdot 10^6$ contains records of calls issued from different locations in Milan \cite{callsdata}\footnote{The complete set of results were omitted due to space restrictions, we note that the results were similar to real dataset $\mathcal{B}_1$}.  We create a synthetic workload $\mathcal{T}$ containing $m = 5\cdot{10^4}$ queries and their answers, i.e., $\{(\mathbf{q},y)_{i}\}_{i=1}^{m}  = \mathcal{T}$. Each query is a 3-d vector $\mathbf{q} = [\mathbf{x}, \theta]$ with answer $y$ where $\mathbf{x} = [x_1,x_2] \in\mathbb{R}^2$ is the center, $\theta\in\mathbb{R}$ is its radius and $y\in\mathbb{R}$ the result against real dataset $\mathcal{B}_1$. We scale $\mathcal{B}_1$ and $\mathcal{T}$ in all their dimensions, restricting their values in $[0,1]$. 
% We must normalize as $\theta$ and $y$ may differ by varying orders of magnitude. 
%As $\theta$ refers to the radius of a given query and the $y/z$ to the result obtained by that query. 
%In cases where $y$ is large, the outcome of (\ref{eq:distance-error}) will only be driven by the prediction error, $(y-\hat{y})^2$ and the value for $|\theta' - u_{k^*l}|$ will be negligible. By normalizing and scaling we ensure that the same weight is given to the query's answer and its radius.
We use workload $\mathcal{T}$ for \textit{Pre-Processing} and \textit{Training} and create a \textit{separate} evaluation set $\mathcal{V}$ containing $|\mathcal{V}| = 0.2\cdot m$ new query-answer pairs. %$(\mathbf{q},y)$, i.e., $\mathbf{q}\in \mathcal{V}: \mathbf{q}\notin \mathcal{T}$.

\textbf{Query Workload ($\mathcal{T},\mathcal{V}$) Distributions:}
It is important to mention that no real-world benchmarks exist for exploratory analytics in general; hence we resort to synthetic workloads. The same fact is recognized in \cite{WWDI17}. For completeness, therefore, we will also use their workload as well in our evaluation. Please note that the workload from \cite{WWDI17} depicts different exploration scenarios, based on zooming-in queries. As such, they represent a particular type of exploration, where queried spaces subsume each other, representing a special-case scenario for the overlaps in the workloads for which XAXA was designed. However, as it follows the same general assumptions we include it for completeness and as a sensitivity-check for XAXA workloads. 
We first generated $\mathcal{T},\mathcal{V}$ with our case studies in mind. We used 2-dim./1-dim. Gaussian and 2-dim./1-dim. Uniform distributions for query locations $\mathbf{x}$ and radius $\theta$ of our queries, respectively. Hence, we obtain four variants of workloads, %each having different multi-modal distributions for the query locations \& radii, 
as shown in Table \ref{tab:distros}. We also list the workload borrowed from \cite{WWDI17} as DC(U).\footnote{$|D|$ is the cardinality of the data-set generated by \cite{WWDI17}.}
\begin{table}
\centering
\caption{Synthetic Query Workloads $\mathcal{T},\mathcal{V}$.}
\begin{tabular}{ |c|c|c| } 
 \hline
 & \textbf{Query Center} $\mathbf{x}$ & \textbf{Radius} $\theta$ \\ \hline
 \textit{Gauss-Gauss} & $\sum_{i=1}^C\mathcal{N}(\mathbf{\mu}_i, \Sigma_i)$ & $\sum_{i=1}^J\mathcal{N}(\mathbf{\mu}_i, \sigma^2_i)$ \\ \hline
 \textit{Gauss-Uni} & $\sum_{i=1}^C\mathcal{N}(\mathbf{\mu}_i, \Sigma_i)$ & $\sum_{i=1}^J U(v_i, v_i+0.02)$ \\ \hline
 \textit{Uni-Gauss} & $\sum_{i=1}^C U(v_i, v_i+0.04)$ & $\sum_{i=1}^J\mathcal{N}(\mathbf{\mu}_i, \sigma^2_i)$ \\ \hline
 \textit{Uni-Uni} & $\sum_{i=1}^C U(v_i, v_i+0.04)$ & $\sum_{i=1}^J U(v_i, v_i+0.02)$ \\ 
 \hline
DC (U)  & $U(0, |D|)$ & $U(0, |D|)$ \\ \hline
\end{tabular}
\label{tab:distros}
\vspace{-13pt}
\end{table}
%The $4$ variants are basically different combinations of the Gaussian and Uniform distributions. 

The parameters $C$ and $J$ signify the number of mixture distributions within each variant. Multiple distributions comprise our workload as we desire to simulate the analysts' behavior issuing queries against the real dataset $\mathcal{B}_1$. For instance, given a number $C\times J$ of analysts, each one of them might be assigned to different geo-locations, hence parameter $C$, issuing queries with different radii, hence $J$, given their objectives. For the parameters $\mathbf{\mu}$ (mean) and $v$ of the location distributions in Table \ref{tab:distros}, we select points uniformly at random ranging in [0,1]. 
The covariance $\Sigma_i$ of each Gaussian distribution, for location $\mathbf{x}$, is set to $0.0001$. 
%As a Data Scientist will most likely issue queries close to the location of interest. 
The number of distributions/query-spaces was set to $C\times J$, where $C \in \mathbb{N}$ and $J \in \mathbb{N}$ with $C=5$ and $J=3$, thus, a mixture of 15 query distributions. The radii covered $2\%-20\%$ of the data space, i.e., $\mu_i$ for $\theta$ was randomly selected in $[0.02, 0.2]$ for Gaussian distributions, with small $\sigma^2$ for the Gaussian distributions, and $v_i$ was in $[0.02, 0.18]$ for Uniform distributions. Further increasing that number would mean that the queries generated would cover a much greater region of the space, thus, making the learning task much easier. The variance $\sigma^2_i$ for  Gaussian distributions of $\theta$ was set to $0.0009$ to leave no overlap with the rest of the $J-1$ distributions. We deliberately leave no overlapping within $\theta$ distributions as we assume there is a multi-modal mixture of distributions with different radii used by analysts. Parameters $z$ and $\alpha$ were set to $z=0.5$ to ensure equal importance during the query assignment step and stochastic gradient learning schedule $\alpha=0.01$.

Our implementation is in Python 2.7. Experiments ran single threaded on a Linux Ubuntu 16.04 using an i7 CPU at 2.20GHz with 6GB RAM. %With the properties of dataset $B$ remaining as described in the previous sections. And 
%We exclude the DaG approach as we found its accuracy to be not acceptable. %In the following paragraphs we describe how an explanation is generated and the way its evaluated.
For every query in the evaluation set $\mathcal{V}$, we take its radius and generate $n$ evenly spaced radii over the interval $[0.02, \theta']$, where $0.02$ is the minimum radius used. For query $q = [\mathbf{x}, \theta]$ we generate and find the answer for $n$ sub-queries, $ST = \{([\mathbf{x}, \theta_1], y_1), \cdots ([\mathbf{x},\theta_n], y_n) \}$. The results of these queries constitute the \textit{Actual Explanation} (AE) and hence the \textit{true} function is represented as a collection of $n$ pairs of sub-radii and $y$, $\{(\theta_i, y_i)\}_{i=1}^n$. The approximated function is the collection of $n$ sub-radii and predicted $\hat{y}$, $\{(\theta_i, \hat{y}_i)\}_{i=1}^n$.
% \paragraph{Data-Gnostic}
% DG fits the obtained $n$ points using MARS and Linear Regression. The returned explanations are linear functions. We can then measure the accuracy of the returned functions by utilizing the predicted $\mathbf{\hat{y}}/\mathbf{\hat{z}}$ compared with the actual $\mathbf{y}$ given by the $n$ sub-queries. 
% \paragraph{LDAG}
% As mentioned, an LDAG explanation directs the query first to an LR and then to all of the associated RRs are engaged. LDAG utilizes the explanation function $f_{k,l}(\theta;\mathbf{w_k})$ to generate $\mathbf{\hat{y}}/\mathbf{\hat{z}}$ which is used to measure the accuracy of LDAG.
% \paragraph{Type of Aggregation Queries}
% The evaluation used \texttt{COUNT} and \texttt{AVERAGE} operators being both popular and representative AQs.

\textbf{Evaluation \& Performance Metrics:}
%\subsubsection*{Expected Explanation Loss Metrics}

\textbf{Information Theoretic:} EEL is measured using the Kullback--Leibler (KL) divergence. Concretely, the result is a scalar value denoting the amount of \textit{information loss} when one chooses to use the approximated explanation function. The EEL with KL divergence is then defined as:
\begin{eqnarray}
\mathcal{L}(f(\theta;\mathbf{x}), \hat{f}(\theta;\mathbf{x})) & = & KL(p(y;\theta,\mathbf{x})||\hat{p}(y;\theta,\mathbf{x})) \\ \nonumber 
& = &\int{p(y;\theta,\mathbf{x})\log{\frac{p(y;\theta,\mathbf{x})}{\hat{p}(y;\theta,\mathbf{x})}}dy},
\end{eqnarray}
with $p(y;\theta,\mathbf{x})$ and $\hat{p}(y;\theta,\mathbf{x})$ being the probability density functions of the true and approximated result, respectively.

\textbf{Goodness of Fit:} The EEL is measured here using the coefficient of determination $R^2$. This metric indicates how much of the \textit{variance} generated by $f(\theta;\mathbf{x})$ can be explained using the approximation $\hat{f}(\theta;\mathbf{x})$. 
This represents the goodness-of-fit of an approximated explanation function over the actual one. It is computed using : $R^2 = 1 - \frac{\sum_i{(y_i - \hat{y}_i)^2}}{\sum_i{(y_i- \overline{y})^2}}$, in which the denominator is proportional to the \textit{variance} of the \textit{true} function and the numerator are the residuals of our approximation. The EEL between 
$f$ and $\hat{f}$ explanations can then be computed as $\mathcal{L}(f(\theta;\mathbf{x}), \hat{f}(\theta;\mathbf{x})) = 1-R^2$. 

\textbf{Model-Based Divergence:} This measures the \textit{similarity} between two explanation functions $f$ and $\hat{f}$
by calculating the inner product of their parameter vectors $\mathbf{a}$ and $\hat{\mathbf{a}}$. The parameter vectors are the \textit{slopes} found using the \textit{n} $(\theta, y)$/$(\theta, \hat{y})$ pairs of the \textit{true} and approximated function respectively. We adopt the \textit{cosine similarity} of the two parameter vectors:
\begin{eqnarray}\label{eq:cosine_similarity}
sim(f(\theta;\mathbf{x}), \hat{f}(\theta;\mathbf{x})) = \frac{\mathbf{a} \cdot \hat{\mathbf{a}}}{\lVert \mathbf{a} \rVert \lVert \hat{\mathbf{a}} \rVert}
\end{eqnarray}
The rates of change/slopes of the \textit{true} and approximated functions (and hence the functions themselves)  are interpreted as being identically similar if: $sim(f(\theta;\mathbf{x}), \hat{f}(\theta;\mathbf{x})) = 1$ and $-1$ if diametrically opposed.
. For translating that metric into $\mathcal{L}$ one would simply negate the value obtained by $\mathcal{L}(f(\theta;\mathbf{x}), \hat{f}(\theta;\mathbf{x})) = 1-sim(f(\theta;\mathbf{x}), \hat{f}(\theta;\mathbf{x}))$.

\textbf{Predictive Accuracy:}
%We have employed two statistical measures to measure the EEL for our proposed explanation function. 
To completely appreciate the accuracy and usefulness of XAXA, we also quantify the predictive accuracy associated with explanation (regression) functions. We employ the \textit{Normalized-Root-Mean-Squared-Error(NRMSE)} for this purpose: $\text{NRMSE} = \frac{1}{y_{max}-y_{min}}(\frac{1}{n}\sum_i^n (y_i - \hat{y}_i)^2)$.
 Essentially, this shows how accurate the results would be if an analyst used the explanation (regression) functions for further data exploration, without issuing DB queries.

% %\subsubsection*{Performance Metrics}

% \textbf{Time Efficiency}: measured using wall clock time in ms. %We compare the time required for LDAG and DG to provide an explanation; \textbf{note:} the time for DG and AE would be the same as the $n$ pairs would also be needed.
\subsection{Experimental Results: Accuracy}
For our experiments we chose to show performance and accuracy results over two representative aggregate functions: \texttt{COUNT} and \texttt{AVG} due to their extensive use and because of their properties, with \texttt{COUNT} being monotonically increasing and \texttt{AVG} being non-linear. 
As baselines do not exist for this kind of problem, we demonstrate the accuracy of our framework with the bounds of the given metrics as our reference. Due to space limitations, we show only representative results.%Thus we can objectively assess our models.
%\subsubsection{Coefficient of Determination}
\begin{figure}
\centering
\begin{tabular}{cc}
  \includegraphics[width=40mm]{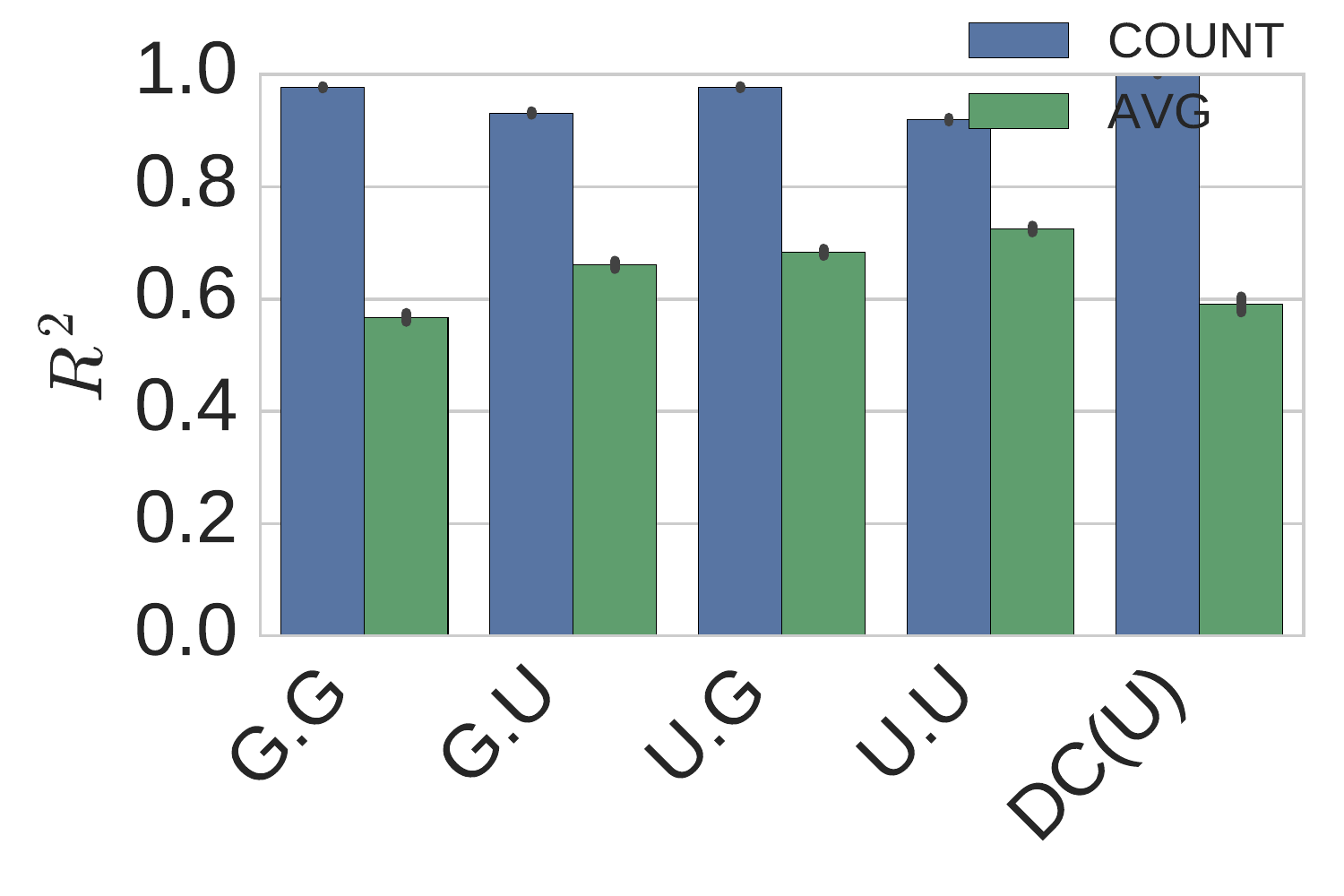}  &
\includegraphics[width=40mm]{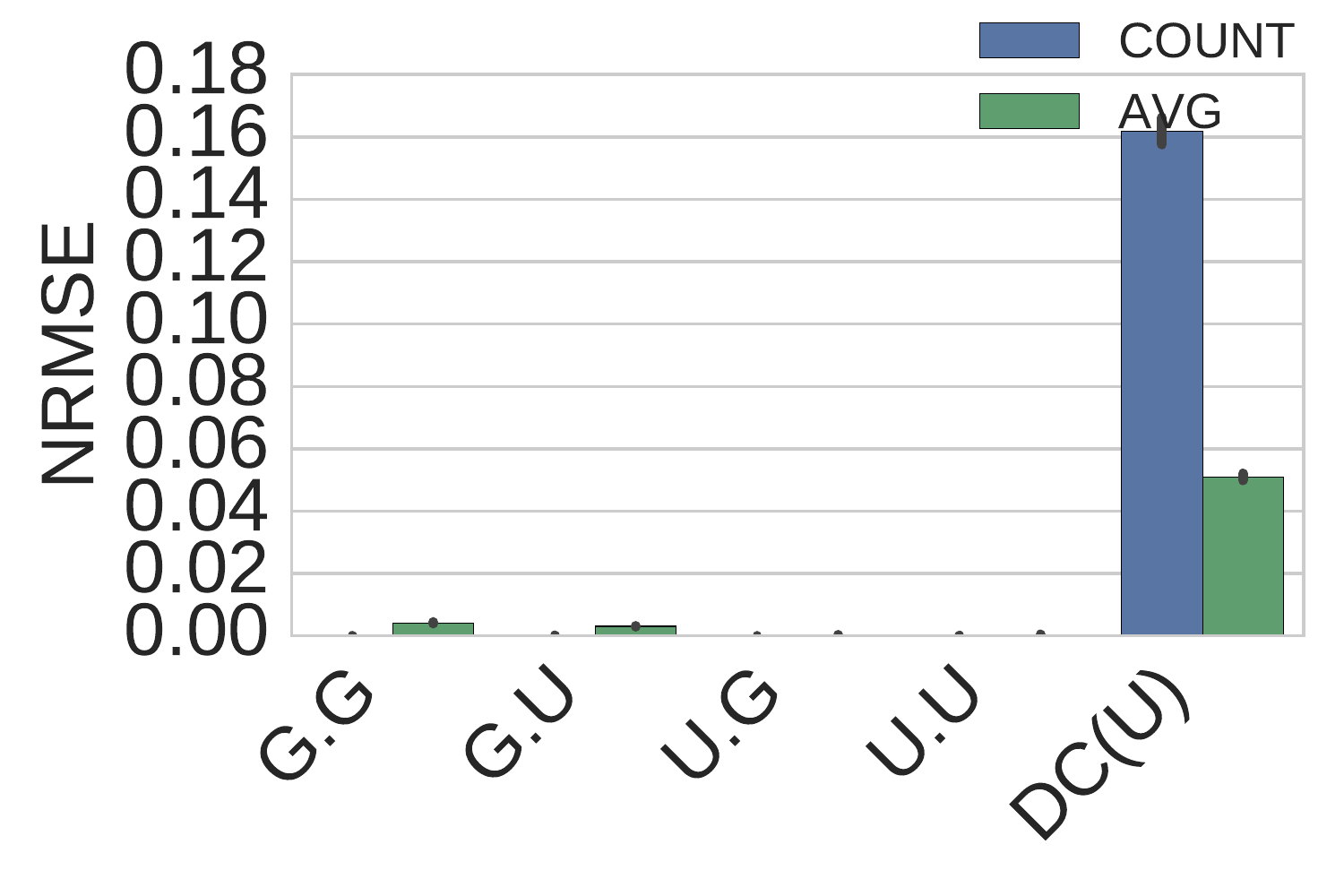}\\
(a) Results for $R^2$. & (b) Results for \textit{NRMSE}.
\end{tabular}
\caption{Goodness of fit and predictive accuracy results.}
\label{fig:results-all}
\end{figure}

\begin{figure}
\centering
%\begin{tabular}{cc}
  \includegraphics[height=3.5cm,width=6.0cm]{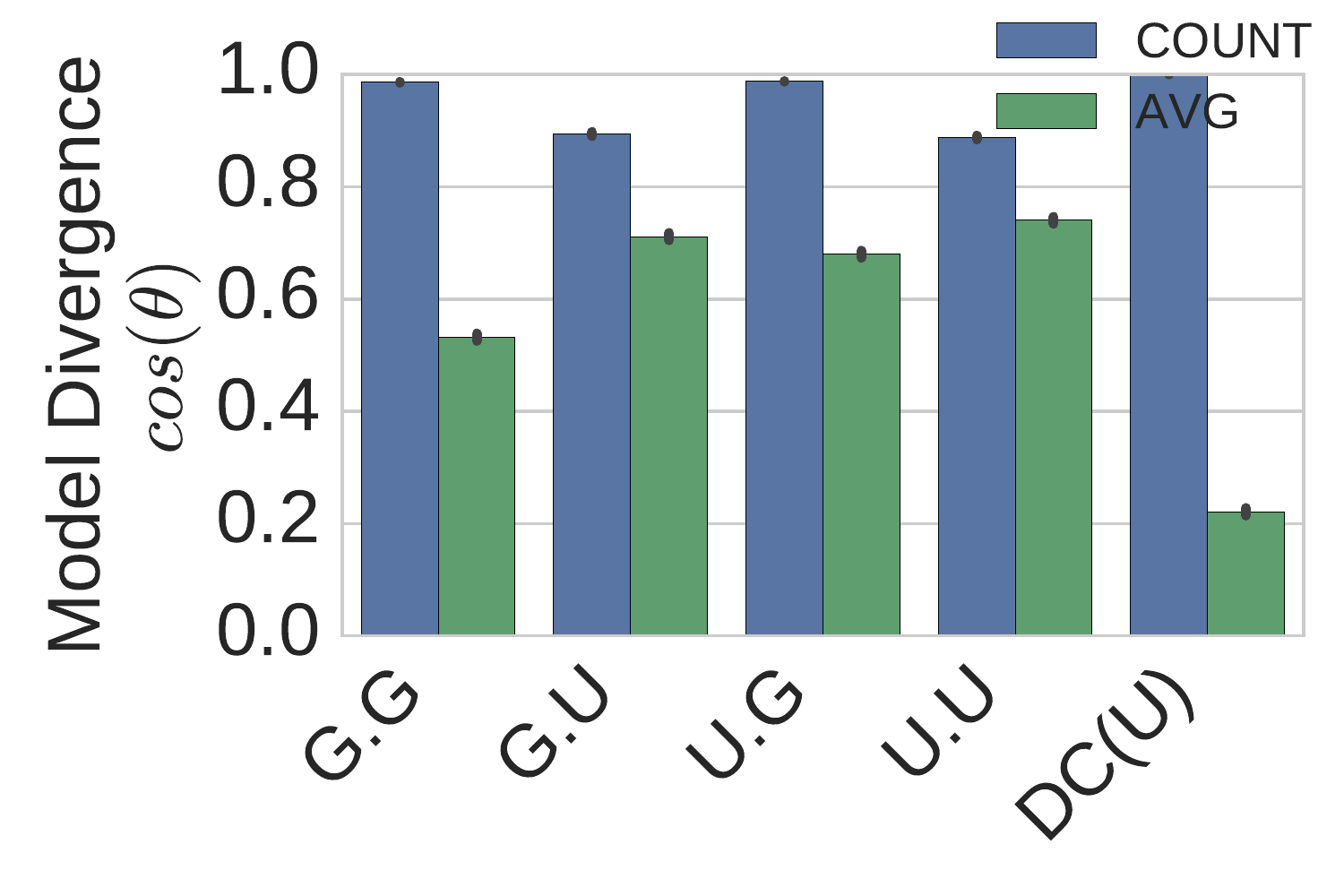}  
%\includegraphics[width=40mm]{barplot-kl}\\
%\end{tabular}
\caption{Results for Model Divergence $\cos(\theta)$.}
\label{fig:results-model-div}
\end{figure}

\begin{figure}
\centering
%\begin{tabular}{cc}
  \includegraphics[height=3.5cm,width=6.0cm]{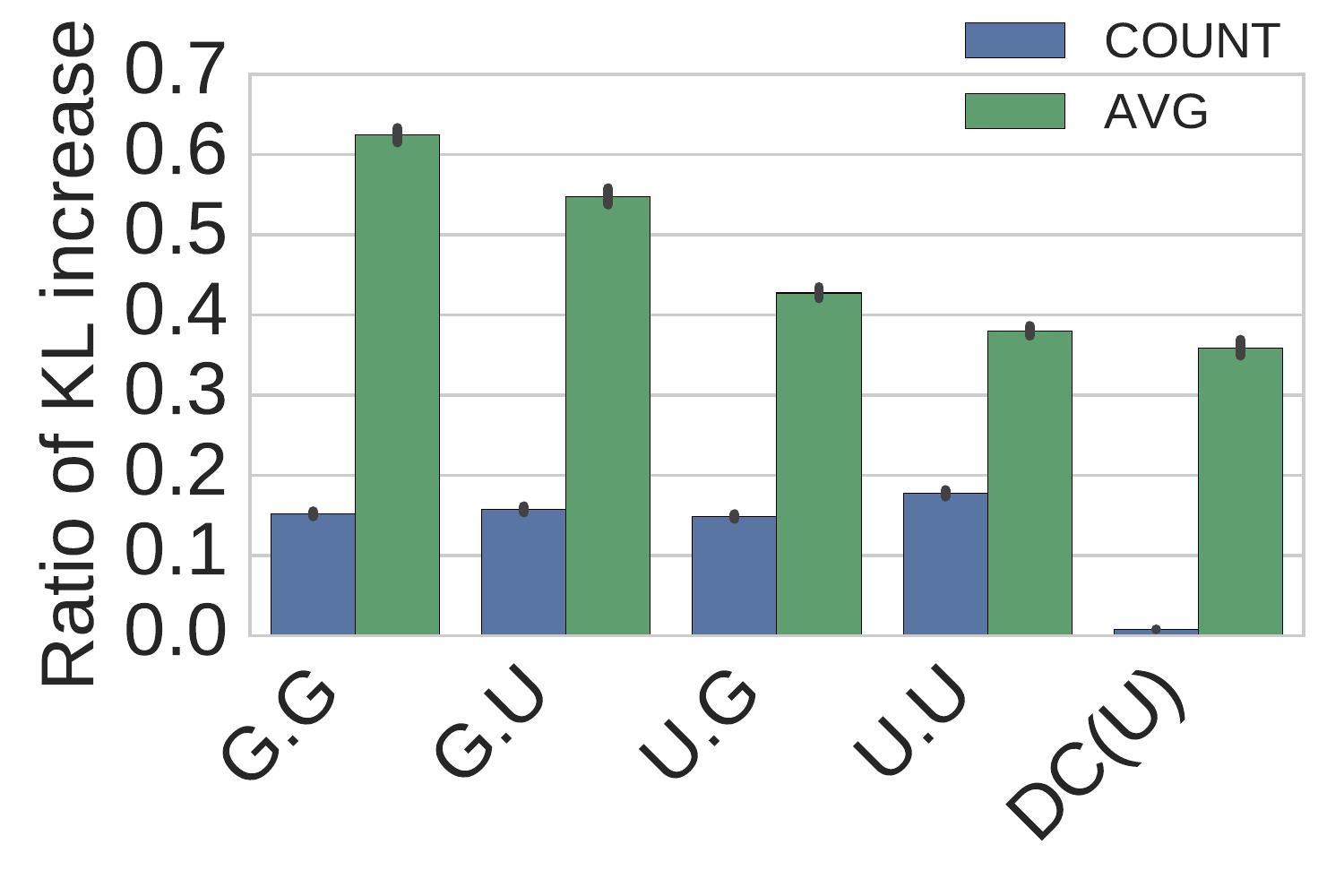}  
%\includegraphics[width=40mm]{barplot-kl}\\
%\end{tabular}
\caption{Results for Information Theoretic metric.}
\label{fig:results-kl}
\end{figure}
\textbf{Goodness of Fit}
Figure \ref{fig:results-all}(a) shows the results for $R^2$ over all workloads for \texttt{AVG} and \texttt{COUNT}. We report on the \textit{average} $R^2$ found by evaluating the explanation functions given for the evaluation set $\mathcal{V}$. Overall, we observe high accuracy across all workloads and aggregate functions. Meaning that our approximate explanation function can explain most of the variation observed by the \textit{true} function. Figure \ref{fig:results-all}(a) also shows the \textit{standard-deviation} for $R^2$ within the evaluation set, which appears to be minimal.
Note that accuracy for \texttt{COUNT} is higher than for \texttt{AVG}; this is expected as \texttt{AVG} fluctuates more than \texttt{COUNT}, the latter being monotonically increasing. 

If the \textit{true} function is highly non-linear, the score for this metric deteriorates as witnessed by the decreased accuracy for \texttt{AVG}. Hence, it should be consulted with care for non-linear functions this is why we also provide measurements for \textit{NRMSE}. Also, it is important to note that if the \textit{true} function is constant, meaning the rate of change is $0$, then the score returned by $R^2$ is $0$. Thus, in cases where such \textit{true} functions exist, $R^2$ should be avoided as it would incorrectly label our approximation as inaccurate. We have encountered many such \textit{true} functions especially for small radii in which no further change is detected.

%\subsubsection{Slope Error}
% In Figure \ref{fig:results-all}(b), we evaluate our explanations based on how well they approximate the changing \textit{rate} (slope) of the \textit{true} function. Again, high accuracy is observed w.r.t how similar the \textit{true} rates are compared to the approximated ones. Reduced similarity is expected in \texttt{AVG} because of its highly varying and sign changing rate. However, similarity is typically above $60\%$ (apart from \textit{DC(U)}).
\textbf{Predictive Accuracy}
As statistical measures for model fitness can be hard to interpret, Figure \ref{fig:results-all}(b) provides results for \textit{NRMSE} when using XAXA explanation (regression) functions for predictions to AQ queries (avoiding accessing the DBMS).
Overall, the predictive accuracy is shown to be excellent for all combinations of Gaussian-Uniform distributed workloads and for both aggregate functions. Even for the worst-case workload of DC(U), accuracy (NRMSE) is below 4\% for \texttt{AVG} and below 18\% for \texttt{COUNT}. As such, an analyst can consult this information to decide on whether to use the approximated function for subsequent computation of AQs to gain speed, instead of waiting idle for an AQ to finish execution.

\textbf{Information Theoretic}
The bar-plot in Figure \ref{fig:results-kl} shows the ratio of bit increase when using the predicted distribution given by the $n$ predictions rather than the $n$ actual values for $y$. This indicates the inefficiency caused by using the approximated explanations, instead of the \textit{true} function. This information, theoretically, allows us to make a decision on whether using such an approximation would lead to the propagation of errors further down the line of our analysis.  We observe that the ratio is high for \texttt{AVG} as its difficulty to be approximated is also evident by our previous metrics. Although some results show the ratio to be high, we note that such inefficiency is tolerable during exploratory analysis in which speed is more important.

\textbf{Model-Based Divergence}
Figure \ref{fig:results-model-div} show the similarity between the true underlying function and the approximation generated by XAXA. Our findings are inline with the rest of the accuracy experiments exhibited in previous sections of our accuracy results. This is an important finding as it shows that the similarity between the shapes of the functions is high. It informs us that the user can rely on the approximation to identify points in which the function behaves non-linearly. If we rely just on the previous metrics this can go unnoticed.

\subsection{Additional Real-World Dataset}
The dataset $\mathcal{B}_2= \{\mathbf{x}\}$ where $\mathbf{x}\in\mathbb{R}^2$ and  $|\mathcal{B}_2| = 5.6 \cdot 10^6$ contains records of calls issued from different locations in Milan \cite{callsdata}. We  show results for workload (\textit{Gauss-Gauss}) on \texttt{Average} and \texttt{COUNT}. %to demonstrate the extensibility of our solution to other datasets and different workloads. The \textit{Data Preparation} and \textit{Experimental Setup} are the same as in our \textit{Experimental Evaluation} section. 
For space reasons we report only on $R^2$. Figure \ref{fig:calls_r2} shows the same patterns as in Figure \ref{fig:r2_eval}, where \texttt{COUNT} is easier to approximate than \texttt{AVG}.
\begin{figure}[!htbp]
\begin{center}
\includegraphics[height=3.5cm,width=7.0cm]{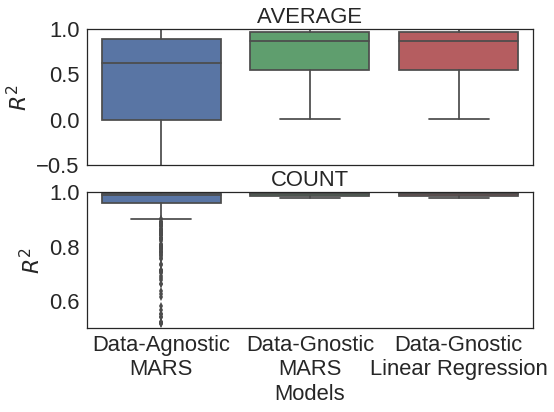}
\vspace{-10pt}
\caption{Coefficient of Determination for $R_2$ (Milan dataset)}
\label{fig:calls_r2}
\end{center}
\end{figure}
Overall, XAXA enjoys high accuracy, learning a number of diverse explanations under many different workloads and datasets. %which makes XAXA attractive for exploratory analytics.
The obtained results for the rest of the metrics and variant workloads are omitted due to space constraints. However, they all follow the same pattern as the results given in our experiments.
%\vspace{-5pt}
\subsection{Experimental Results: Efficiency and Scalability}
% \textit{Time} was measured in milliseconds. As expected, LDAG outperforms DG by at least two orders of magnitude as queries are not executed against the DBMS. We have also performed experiments showcasing the scalability of our methods measuring both the \textit{Training} and \textit{Pre-Processing} overheads. The results showed, both stages only take at most 2-3 minutes to execute, noting that the \textit{Training} phase happens on-line. Also, we tested our approach with an ever increasing data set, noting that even though the DG approach has an exponentially degrading performance the LDAG approach remained constant. Figures for this section are omitted due to space limitations; therefore, LDAG is a highly accurate and scalable, which with minimal overhead serves as an explanation mechanism.
\label{appendix:efficiency}
\subsubsection{XAXA is invariant to data size changes}
%Increasing data sizes have a drastic impact on DG, as evident in figure \ref{fig:increase_data}. %This is due to the fact that DG requires access to the underlying data to generate an explanation function $f(\theta;\mathbf{x})$. 
XAXA remains invariant to data size changes as shown in Figure \ref{fig:scalability-xaxa} (b). Dataset $\mathcal{B}_1$ was scaled up by increasing the number of data points in $\mathcal{B}_1$ respecting the original distribution. As a reference we show the time needed to compute aggregates for $\mathcal{B}_1$ using a DBMS/Big Data engine.

\subsubsection{Scalability of XAXA}
We show that XAXA scales linearly with some minor fluctuations attributed to the inner workings of the system and different libraries used. Evidently, both \textit{Pre-Processing} stage and \textit{Training} stage do not require more than a few minutes to complete. We note that \textit{Training} stage happens on-line, however we test the stage off-line to demonstrate its scalability; results are shown in Figure \ref{fig:scalability-xaxa}(a).

\begin{figure}

\begin{tabular}{ll}
  \includegraphics[width=40mm]{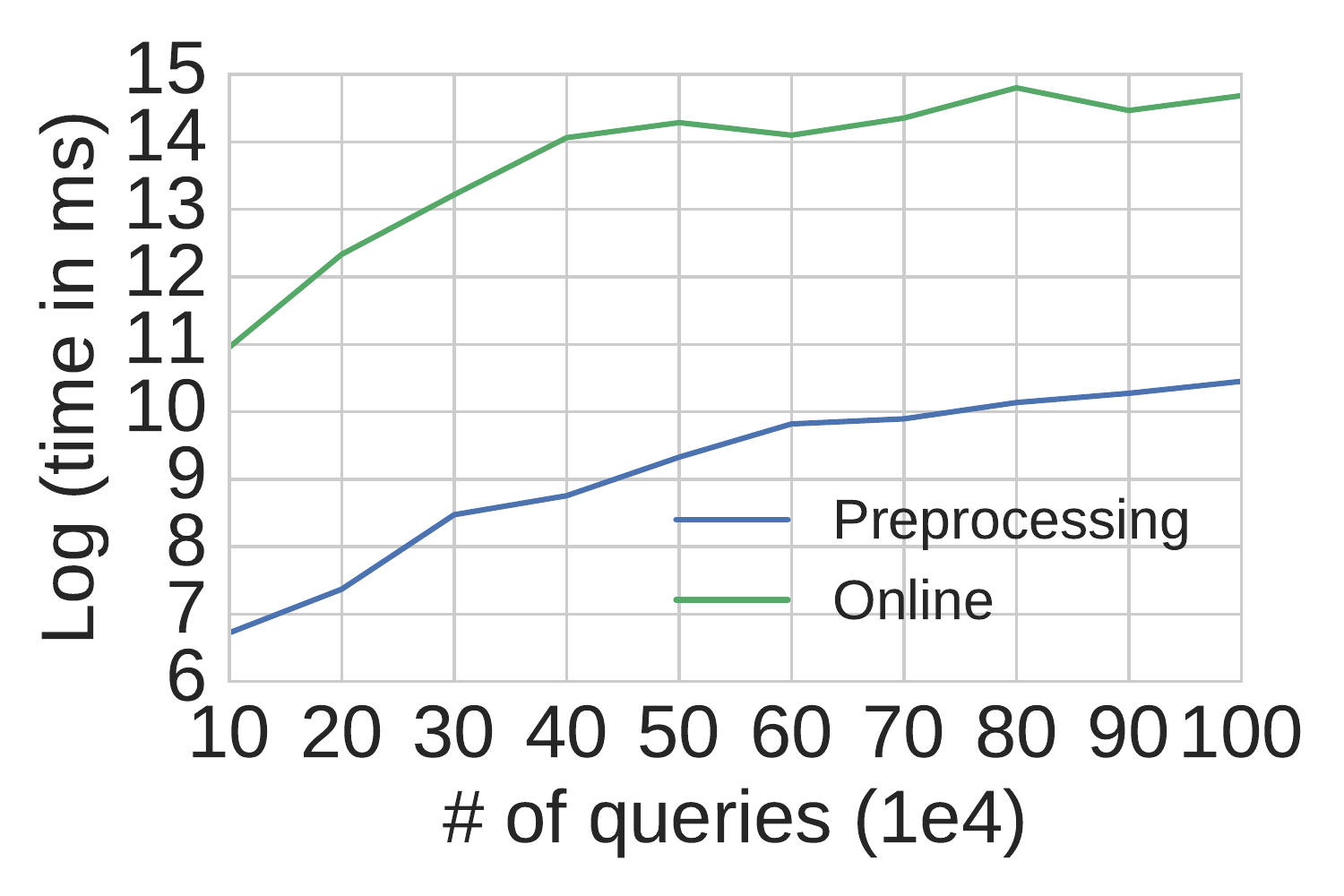} &   \includegraphics[width=38mm]{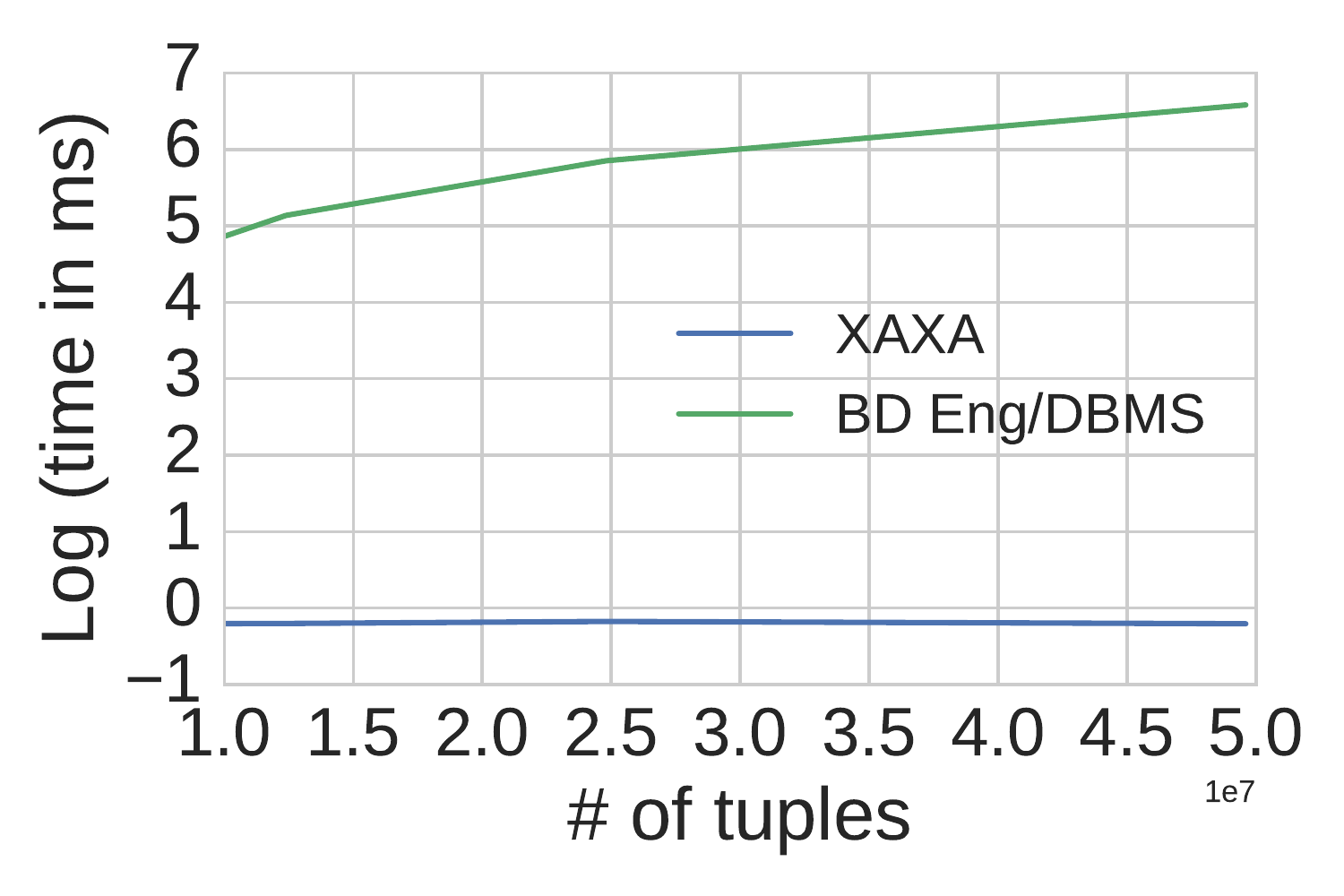} \\
(a) \textit{Pre-Processing} \& \textit{Training}. & (b) Time vs. $|\mathcal{B}_1|$ increase. \\[6pt]
\end{tabular}
%\vspace{-10pt}

\caption{Scalability of XAXA in terms of explanation time and data set size.}
\label{fig:scalability-xaxa}
\end{figure}

\subsection{Experimental Results: Sensitivity to Parameters}
We evaluated the sensitivity on central parameters $K$ and $L$. To do that we measure the \textit{SSQE} defined in (\ref{eq:ssqe}) and $R^2$ using using \textit{Gauss-Gauss} workload over $\mathcal{B}_1$. Results are shown in Figure \ref{fig:sensitivity-parameters}.
% All of the measurements were taken separately for parameters $K$ and $L$, with one parameter set to a pre-defined value when evaluating for the other parameter. 
As evident SSQE is a good indication of the accuracy of our model thus sufficiently minimizing it results in better explanations. We also, show the effectiveness of our heuristic algorithm in finding a near-optimal parameter for $K$ and $L$  by showing the difference in SSQE between successive clustering (bottom plot). 
Thus, by setting our threshold $\epsilon$ at the value indicated by the vertical red line in the graph, we can get near-optimal parameters. 
Figure \ref{fig:sensitivity-parameters} shows that once $K$ is large enough, sensitivity is minimal. However, small values for $K$ are detrimental and values below $5$ are not shown for visual clarity. As for $L$ we note that in this case the sensitivity is minimal with our algorithm still being able to find a near-optimal value for $L$. We do note, that in other cases where the distributions for $L$ are even further apart sensitivity is increased. We omit those results due to space constraints. 
%\vspace{-25pt}
\begin{figure}[!htbp]
\begin{center}
\includegraphics[height=2.5cm,width=5.5cm]{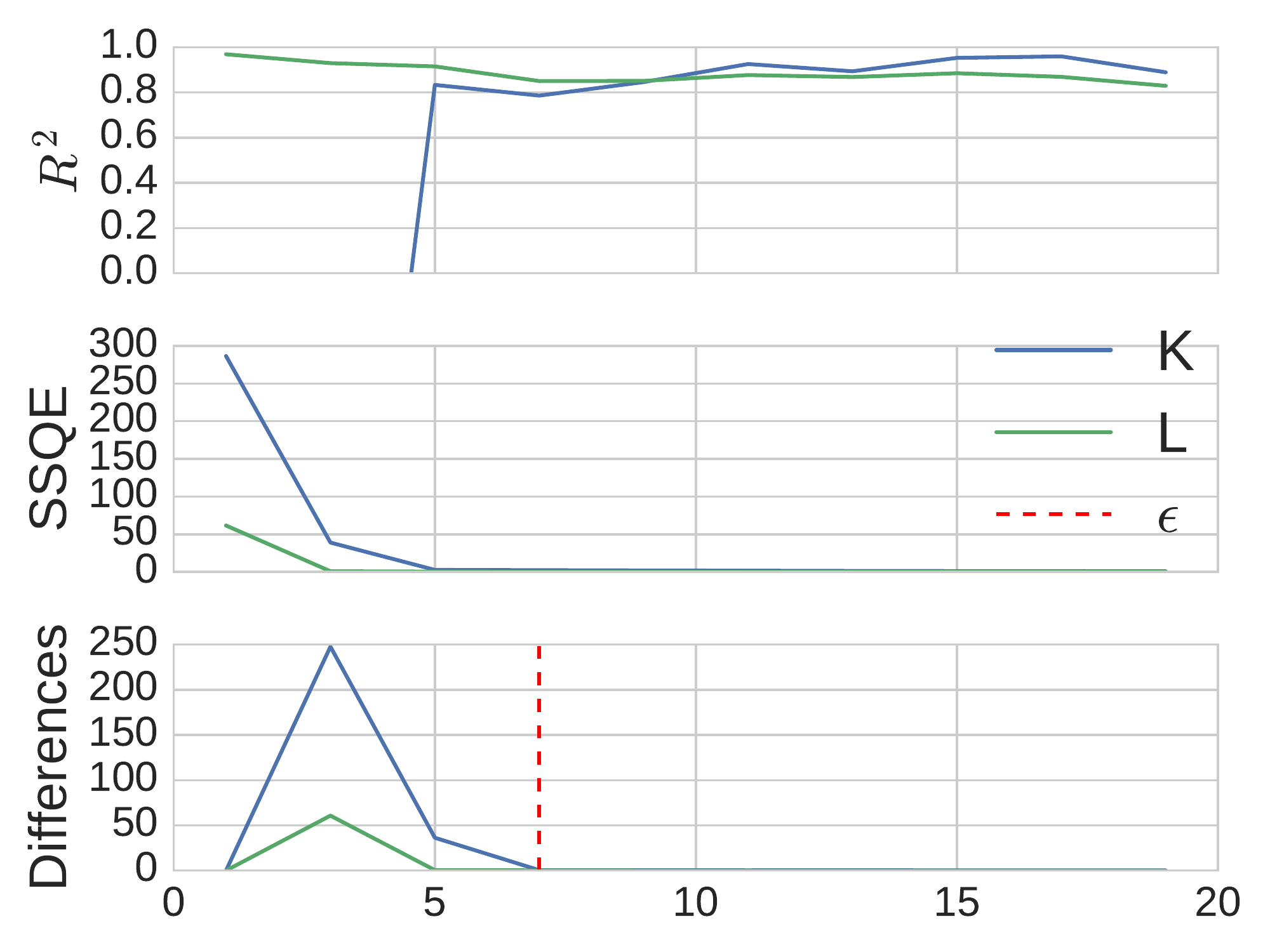}
%\vspace{-10pt}
\caption{Sensitivity to parameters $K$ and $L$. Parameter $e$ is the pre-defined threshold.}
\label{fig:sensitivity-parameters}
\end{center}
%\vspace{-15pt}
\end{figure}
%\vspace{-10pt}
%\vspace{-10pt}
\section{Conclusions}
We have defined a novel class of explanations for Aggregate Queries, which are of particular importance for exploratory analytics. The proposed AQ explanations are succinct, assuming the form of functions. As such, they convey rich information to analysts about the queried data subspaces and as to how the aggregate value depends on key parameters of the queried space. Furthermore, they allow analysts to utilize these explanation functions for their explorations without the need to issue more queries to the DBMS.
% Thus, given the importance of statistical analyses and exploratory analytics, these explanations can significantly empower in-DBMS analytics for data exploration tasks. %helping analysts to understand data subspaces and analytical query results, and guiding them in their data explorations.
We have formulated the problem of deriving AQ explanations as a joint optimization problem and have provided novel statistical/machine learning models for its solution. The proposed scheme for computing explanations does not require DBMS data accesses ensuring efficiency and scalability. %Novel ML models and algorithms were provided to compute approximate explanations, trained by previous AQs and their results, while ensuring high accuracy. 
%Future work includes applying and extending our results for other AQ types and dealing with on-line learning techniques to accommodate data and query pattern changes.
%\end{document}  % This is where a 'short' article might terminate

\section{Acknowledgments}
This work is funded by the EU/H2020 Marie Sklodowska-Curie (MSCA-IF-2016) under the INNOVATE project; Grant$\#$745829. Fotis Savva is funded by an EPSRC scholarship, grant number 1804140. Christos Anagnostopoulos and Peter Triantafillou were also supported by EPSRC grant EP/R018634/1

\bibliographystyle{abbrv}
\bibliography{bibliography}

%Appendix A

\section{Additional Metrics}
\subsection{Alternative to Coefficient of Determination}
It has been noted that $R^2$ is susceptible to outliers \cite{legates1999evaluating}. Therefore we also use an alternative to $R^2$ suggested in Legates \emph{et al.} which uses the absolute value instead of the squared value. It has the following form
\begin{equation}\label{eq:r2_A}
R^2 = 1 - \frac{|y - \hat{y}|}{|y - \overline{y}|}
\end{equation}
and it used to get a more accurate estimate for $R^2$ in cases where outliers have huge impact on its value.
\subsubsection{Result for Alternative to $R^2$}
\begin{figure}[!htbp]
\begin{tabular}{cc}
  \includegraphics[width=40mm]{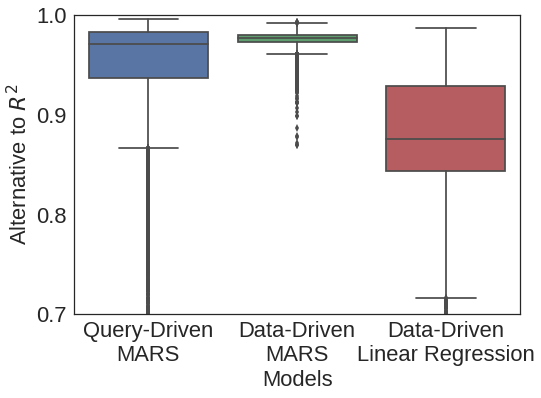} &   \includegraphics[width=40mm]{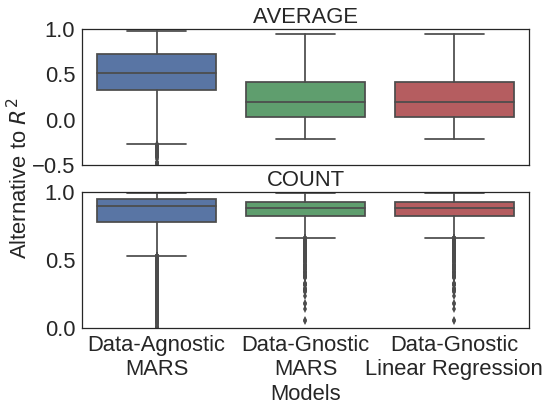} \\
(a) Gaussian $\mathbf{x}$ - Gaussian $\theta$ & (b) Gaussian $\mathbf{x}$ - Uniform $\theta$ \\[6pt]
 \includegraphics[width=40mm]{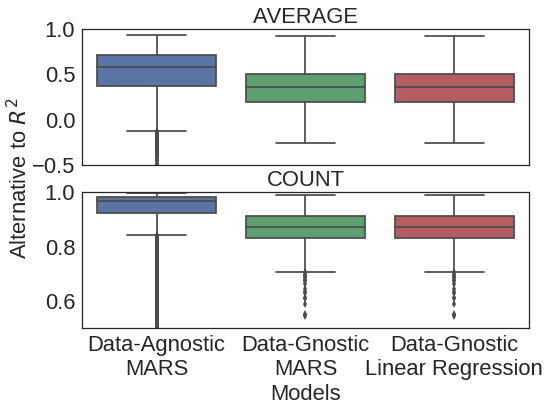} &   \includegraphics[width=40mm]{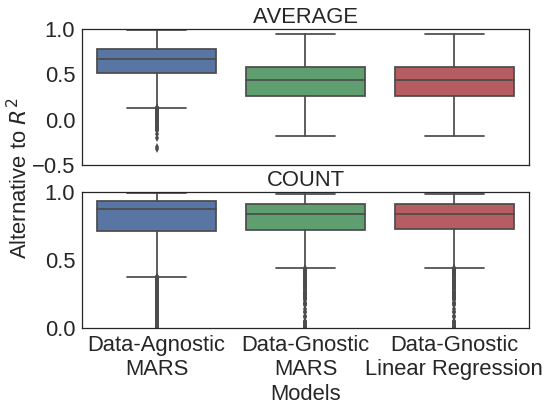} \\
(c) Uniform $\mathbf{x}$ - Gaussian $\theta$ & (d) Uniform $\mathbf{x}$ - Uniform $\theta$ \\[6pt]
\end{tabular}
\caption{Results for Alternative to Coefficient of Determination metric}
\label{fig:eval_alt_r2}
\end{figure}
Examining the generated boxplots, figure \ref{fig:eval_alt_r2}, for the Alternative to $R^2$ and comparing with the ones for $R^2$ we see that the results slight decreased across all models, aggregation operators and datasets. This is due to the fact that $R^2$ tends to inflate the results.
\paragraph{Main point}
Results observed are different than $R^2$ but this is expected and the behavior of the models across both aggregation operators and datasets is the same.

% \subsection{Increasing number of LRs increases accuracy}
% By increasing the number of LRs/clusters in the query-space we effectively increase accuracy as the level of quantization becomes more fine-grained and incoming queries get mapped to LRs that are closer in proximity.
% \begin{figure}
% \begin{center}
% \includegraphics[height=3.5cm,width=8.0cm]{number_of_clusters}
% \caption{Increasing the number of LRs/clusters in the query-space there is an increase accuracy. The y-axis denotes the $R^2$ measured for DAG having the number of LRs specified in the x-axis.}
% \label{fig:number_of_clusters}
% \end{center}
% \end{figure}

\section{Proofs of Theorems}
\subsection{Proof of Theorem \ref{theorem:1}}
\label{appendix:theorem:1}
We adopt the Robbins-Monro stochastic approximation for minimizing the combining distance-wise and prediction-wise loss given a L2 RR $u$, that is minimizing $\mathcal{E}(u) = \mu |\theta - u| + (1-\mu)(y - \hat{f}(\theta;\mathbf{w}))^{2}$, using SGD over $\mathcal{E}(u)$.
Given the $t$-th training pair $(\mathbf{q}(t), y(t))$, 
the stochastic sample $E(t)$ of $\mathcal{E}(u)$ has to decrease at each new pair at $t$ by descending in the direction of its negative gradient with respect to $u(t)$. 
Hence, the update rule for L2 RR $u$ is derived by: 
\[
\Delta u(t) = -\alpha(t) \frac{\partial E(t)}{\partial u(t)},
\]
where scalar $\alpha(t)$ satisfies $\sum_{t=0}^{\infty}\alpha(t) = \infty$ and $\sum_{t=0}^{\infty}\alpha(t) < \infty$. 
From the partial derivative of $E(t)$ 
we obtain $\Delta u \leftarrow \alpha \mu \mbox{sgn}(\theta - u)$. 
By starting with arbitrary initial training pair $(\mathbf{q}(0), y(0))$, 
the sequence $\{u(t)\}$ converges to optimal 
L2 RR $u$ parameters. 

\subsection{Proof of Theorem \ref{theorem:2}}
\label{appendix:theorem:2}
Consider the optimal update rule in (\ref{eq:sgd}) for L2 RR $u$
and suppose that $u$ has reached equilibrium, i.e., $\Delta u = 0$ holds with probability 1. By taking the expectations of both sides and replacing $\Delta u$ with the update rule from Theorem \ref{theorem:1}:
\begin{eqnarray*}
\mathbb{E}[\Delta u] & = & \int_{\mathbb{R}}\mbox{sgn}(\theta- u)p(\theta)\mbox{d}\theta =  P(\theta \geq u) \int_{\mathbb{R}}p(\theta)\mbox{d}\theta \\
& - & P(\theta < u) \int_{\mathbb{R}}p(\theta)\mbox{d}\theta = 2P(\theta \geq u)-1.
\end{eqnarray*}
Since $\Delta u = 0$ thus $u$ is constant, then $P(\theta \geq u) = \frac{1}{2}$, which denotes that $u$ converges to the median of radii for those queries represented by L1 RL and the associated L2 RR. 

\section{Near Optimal $k$ for $K$-means}
\label{algo:near-optimal}
We provide an algorithm for near-optimal $k$-means we note that there are multiple such algorithms available in the literature \cite{hamerly2004learning}. However, this is not part of our focused work thus a simple solution was preferred to alleviate such problem. In addition, our work can be utilized with any kind of Vector-Quantization algorithm thus $k$-means was chosen because of its simplicity and wide use.
\begin{algorithm}[!htb]
\caption{Estimating a near-optimal $K$}
\begin{algorithmic}[H]
\State Input: $\epsilon$; $K$ \Comment{initial $K$; predefined improvement threshold.}
\State $\mathcal{W} = \emptyset$; $\mathcal{X} = \{\mathbf{x}_{i}\}_{i=1}^{m}: \mathbf{x} \in \mathcal{T}$ \Comment{set of LRs; query centers $\mathcal{T}$}
% \State $\mathbf{X} \gets \{\mathbf{x}| \mathbf{x}\subset\mathbf{q} \}  \forall\mathbf{q}$
% \State n $\gets |X|$
% \State prevSSE $\gets -1$
\While{\texttt{TRUE}}
	\State $\mathcal{W} \gets \mathit{KMeans(K,\mathcal{X})}$
    \Comment{call K-Means algorithm with $K$ LRs}
    \State SSQE $\gets \sum_{i=1}^{m}\min\limits_{\mathbf{w_k} \in \mathcal{W}}(\lVert\mathbf{x}_{i}-\mathbf{w}_{k}\rVert^2)$
    \Comment{Calculate SSQE}
    \If{$\Delta |SSQE| > \epsilon$} \Comment{improvement}
        \State $K \gets K + 1$ \Comment{increase $K$} 
    \Else 
        \State \textbf{break} \Comment{no more improvement; exit}
    \EndIf
\EndWhile
\State Return: $\mathcal{W}$ \Comment{set of $K$ LRs}
\end{algorithmic}
\label{alg:algorithm_near}
\end{algorithm}
The algorithm is fairly straight-forward and is in-line with a method called the \"Elbow Method\" in approximating an optimal $K$ for $K$-means. It essentially performs several passes over the data, applying the algorithm and obtaining an SSQE. It stops until a pre-defined threshold has been reached.

\end{document}